\documentclass[11pt, oneside]{article}   	
\usepackage{geometry}                		
\usepackage{graphicx}				
								
\usepackage{enumerate}
\usepackage{hyperref}
\usepackage{subcaption}
\usepackage{amssymb}
\usepackage{amsmath}

\usepackage{amsthm}
\usepackage[linesnumbered,ruled,vlined]{algorithm2e}
\usepackage{environ}

\newcounter{parentAlgoLine}
\SetKwBlock{BEGIN}{}{}

\makeatletter
\NewEnviron{substeps}{%
  \protected@edef\theparentequation{\theequation}%
  \setcounter{parentAlgoLine}{\value{AlgoLine}}%
  \setcounter{AlgoLine}{0}%
  \BEGIN{\BODY}%
  \setcounter{AlgoLine}{\value{parentAlgoLine}}%
  \ignorespacesafterend
}
\makeatother

\SetKwInput{KwInput}{Input}                
\SetKwInput{KwOutput}{Output}



\newtheorem{fact}{Fact}

\newtheorem{proposition}{Proposition}

\theoremstyle{definition}
\newtheorem{definition}{Definition}

\def\Re{\textnormal{Re}}
\def\Im{\textnormal{Im}}

\def\real{\mathbb{R}}

\def\complex{\mathbb{C}}
\newcommand{\code}[1]{\texttt{#1}}

\newcommand{\norm}[1]{\|#1\|}
\def\qiskitdynamics{Qiskit Dynamics}
\def\perturbfunc{\code{solve\_lmde\_perturbation}}

\def\symD{\mathcal{D}}
\def\symE{\mathcal{E}}
\def\symI{\mathcal{I}}
\def\symQ{\mathcal{Q}}
\def\symO{\mathcal{O}}

\DeclareMathOperator\erf{erf}

\usepackage{bbm}
\usepackage{dsfont}
\newcommand\id{\mathds{1}}

\DeclareFixedFont{\ttb}{T1}{txtt}{bx}{n}{10} 
\DeclareFixedFont{\ttm}{T1}{txtt}{m}{n}{10}  

\usepackage{color}
\definecolor{deepblue}{rgb}{0,0,0.5}
\definecolor{deepred}{rgb}{0.6,0,0}
\definecolor{deepgreen}{rgb}{0,0.5,0}
\usepackage{listings}
\definecolor{commentgreen}{rgb}{0.25,0.5,0.5}

\newcommand\pythonstyle{\lstset{
language=Python,
basicstyle=\ttm,
tabsize=4,
morekeywords={self},              
commentstyle=\color{commentgreen},
keywordstyle=\ttb\color{deepblue},
emph={MyClass,__init__},          
emphstyle=\ttb\color{blue},    
stringstyle=\color{deepred},
frame=tb,                         
showstringspaces=false
}}

\lstnewenvironment{python}[1][]
{
\pythonstyle
\lstset{#1}
}
{}

\newcommand\pythoninline[1]{{\pythonstyle\lstinline!#1!}}

\usepackage[dvipsnames]{xcolor}
\usepackage{color,soul}

\title{Algorithms for perturbative analysis and simulation of quantum dynamics}

\usepackage{authblk}

\author[1,*]{Daniel Puzzuoli}
\author[2]{Sophia Fuhui Lin}
\author[3]{Moein Malekakhlagh}
\author[3]{Emily Pritchett}
\author[3]{Benjamin Rosand}
\author[3]{Christopher J. Wood}

\def\affilsize{\small}

\affil[1]{\affilsize IBM Quantum, IBM Canada, Markham, ON, L3R 9Z7, Canada}
\affil[2]{\affilsize Department of Computer Science, University of Chicago, Chicago, IL, 60615, USA}
\affil[3]{\affilsize IBM Quantum, IBM T.J. Watson Research Center, Yorktown Heights, NY, 10598, USA}
\affil[*]{\affilsize \href{mailto:daniel.puzzuoli1@ibm.com}{daniel.puzzuoli1@ibm.com}}

\begin{document}
\maketitle

\begin{abstract}
We develop general purpose algorithms for computing and utilizing both the Dyson series and Magnus expansion, with the goal of facilitating numerical perturbative studies of quantum dynamics. To enable broad applications to models with multiple parameters, we phrase our algorithms in terms of multivariable sensitivity analysis, for either the solution or the time-averaged generator of the evolution over a fixed time-interval. These tools simultaneously compute a collection of terms up to arbitrary order, and are general in the sense that the model can depend on the parameters in an arbitrary time-dependent way. We implement the algorithms in the open source software package \qiskitdynamics{}, utilizing the JAX array library to enable just-in-time compilation, automatic differentiation, and GPU execution of all computations. Using a model of a single transmon, we demonstrate how to use these tools to approximate fidelity in a region of model parameter space, as well as construct perturbative robust control objectives. 

We also derive and implement Dyson and Magnus-based variations of the recently introduced \code{Dysolve} algorithm [Shillito et al., \emph{Physical Review Research}, 3(3):033266] for simulating linear matrix differential equations. We show how the pre-computation step can be phrased as a multivariable expansion computation problem with fewer terms than in the original method. When simulating a two-transmon entangling gate on a GPU, we find the Dyson and Magnus-based solvers provide a speedup over traditional ODE solvers, ranging from roughly $2\times$ to $4\times$ for a solution and $10\times$ to $60\times$ for a gradient, depending on solution accuracy.
\end{abstract}

\section{Introduction}

Accurate and high-performance simulation of the physics of quantum systems is a key component in device and control engineering workflows in quantum computation. Due to the complexity of the time-dependent differential equations involved, perturbative techniques are commonly employed to simplify models of these systems. The Dyson series \cite{dyson_radiation_1949} and Magnus expansion \cite{magnus_exponential_1954, blanes_magnus_2009} are widely used time-dependent perturbation theory tools for studying quantum system dynamics. As perturbative expansions, they enable \emph{sensitivity analysis}: the quantification of how the dynamics change with local changes to system parameter values. They are utilized heavily in theoretical studies of quantum systems, as they enable perturbative extension of analytically tractable problems through the construction of ``effective Hamiltonians'' \cite{soliverez_general_1981,james_effective_2011,bravyi_schriefferwolff_2011,zeuch_exact_2020}. For example, in superconducting quantum computing, effective Hamiltonian models have been derived for microwave-activated single-qubit \cite{gambetta_analytic_2011} and two-qubit gates, such as the cross-resonance gate \cite{magesan_effective_2020,malekakhlagh_first-principles_2020}. In the field of quantum control, starting with Average Hamiltonian Theory \cite{haeberlen_1968}, these expansions have also been used to design open-loop \emph{robust} control sequences that suppress the effect of perturbations \cite{waugh_1968, mansfield_1971, rhim_1973, mehring_1983, takegoshi_1985,levitt_1986, cory_1991}, sequences for sensing by \emph{enhancing} perturbations \cite{cory_1990}, dynamical decoupling \cite{viola_1998,viola_1999,khodjasteh_2005}, dynamically corrected gates \cite{khodjasteh_2009, khodjasteh_2010}, the filter function formalism \cite{green_2013, pazsilva_2014}, and Magnus-based methods for suppressing non-adiabatic errors \cite{ribeiro_systematic_2017}, among others.

More recent work has focused on the numerical computation of these expansions, often with the goal of generalizing the previously mentioned applications to contexts which are not analytically tractable. For robust control, numerical methods have been developed to compute expansion terms in the \emph{interaction frame} \cite{evans_timedependent_1967,haeberlen_1968,mehring_1983} of a control sequence, e.g. the Dyson-like terms of \cite{haas_engineering_2019,tabatabaei_numerical_2020} and filter function formalism terms \cite{cywinski_how_2008,green_arbitrary_2013,hangleiter_filter-function_2021}, the latter of which have been implemented in software packages \cite{ball_software_2021,tobias_hangleiter_2022_6574304}. In the context of simulation of quantum systems, Ref. \cite{shillito_fast_2020} gives an algorithm for computing Dyson-series-derived expressions as part of the \emph{pre-computation} step of the \code{Dysolve} algorithm, with similar methods appearing in \cite{kalev_integral-free_2020}.

In this paper we develop software tools for numerically computing and utilizing the Dyson series and Magnus expansion. The general goal is to define and implement algorithms for broad versions of these computational problems, so that they may be used as primitives in many numerical research applications, including robust control and classical simulation of quantum systems. By making the construction of these terms easily accessible, their usefulness in new and existing methods can be explored to higher orders and applied to models with more parameters.

Towards this end, this paper is organized as follows:
\begin{itemize}
	\item Section \ref{section:multivariable} defines the multivariable Dyson series and Magnus expansion, and generally introduces the power series and interaction frame notation used throughout the paper.
	\item Section \ref{section:algorithms} gives algorithms for computing both multivariable Dyson series and Magnus expansion terms, analyzes their scaling, and describes the implementation in the Qiskit Dynamics package.
	\item Section \ref{section:control} demonstrates the usage of the software tools in an example of a robust control problem. In particular, it is shown how the tools can be used to approximate gate fidelity in a region of model parameter space, as well as to construct robustness objectives incorporating higher order expansion terms.
	\item Section \ref{section:numerical_integrators} derives Dyson and Magnus-based variants of the \code{Dysolve} algorithm \cite{shillito_fast_2020} for solving linear matrix differential equations, and describes an implementation of these algorithms in Qiskit Dynamics, building on the implementation described in Section \ref{section:algorithms}. We benchmark these solvers against the standard ODE solvers available in Qiskit Dynamics for the problem of simulating a two transmon CR gate. We find a speed up over traditional solvers in Qiskit Dynamics of roughly $2\times$ to $4\times$ for a solution and $10\times$ to $60\times$ for a gradient on GPU, depending on the accuracy of the solutions.
\end{itemize}

The paper ends with a discussion in Section \ref{section:discussion}.
\section{Multivariable Dyson series and Magnus expansion} \label{section:multivariable}

We begin by informally introducing the computational problems we consider in this paper. Consider a linear matrix differential equation (LMDE):
\begin{equation} \label{equation:informal_lmde}
	\dot{U}(t, c_0, \dots, c_{r-1}) = G(t, c_0, \dots, c_{r-1})U(t, c_0, \dots, c_{r-1}),
\end{equation}
where $t$ is time, and $c_0, \dots, c_{r-1}$ are some parameters of the generator $G$ and solution $U$, which are assumed to be perturbative. Note that the Schrodinger equation is an LMDE under the association $G=-iH$, for $H$ the Hamiltonian, along with other common master equations, such as the Lindblad and Bloch-Redfield equations.\footnote{The Lindblad and Bloch-Redfield equations, which are differential equations for density matrix evolution, are not typically presented in the form of Equation \eqref{equation:informal_lmde}. However, their right-hand sides are linear functions of the density matrix, and therefore they can be rewritten in the form of Equation \eqref{equation:informal_lmde} using a vectorization convention.} 

Given an integration interval $[t_0, t_f]$ and a power-series decomposition of $G$ in the parameters $c_0, \dots, c_{r-1}$ (centred at $0$):
\begin{equation}
	G(t, c_0, \dots, c_{r-1}) = G_\emptyset(t) + \sum_{k=1}^\infty \sum_{0 \leq i_1 \leq \dots \leq i_k \leq r-1}c_{i_1} \dots c_{i_k} G_{(i_1, \dots, i_k)}(t), \label{equation:generator_power_series}
\end{equation}
with the functions $G_\emptyset$ and $G_{(i_1, \dots, i_k)}$ being \emph{completely arbitrary} user-defined functions, the problem is to compute corresponding terms in the truncated power series for the solution $U$ itself:
\begin{equation}
	U(t_0, t_f) = U_\emptyset + \sum_{k=1}^\infty \sum_{0 \leq i_1 \leq \dots \leq i_k \leq r-1}c_{i_1} \dots c_{i_k} U_{(i_1, \dots, i_k)},
\end{equation}
or for the time-averaged generator $\Omega$: 
\begin{equation}
	\Omega = \Omega_\emptyset + \sum_{k=1}^\infty \sum_{0 \leq i_1 \leq \dots \leq i_k \leq r-1}c_{i_1} \dots c_{i_k} \Omega_{(i_1, \dots, i_k)},
\end{equation}
implicitly defined to satisfy $U(t_0, t_f) = \exp(\Omega)$.\footnote{These tasks can equivalently be phrased as performing sensitivity analysis, or numerically finding a series solution, albeit specialized to the case of a linear matrix differential equation.}

While we phrase the time-dependent operators $G_\emptyset(t)$ and $G_{(i_1, \dots, i_k)}(t)$ in Equation \eqref{equation:generator_power_series} as being arbitrary from a computational generality perspective, they are fixed by the parameterization of the generator $G(t, c_0, \dots, c_{r-1})$. That is, it holds that $G_\emptyset(t) = G(t, 0, \dots, 0)$ (i.e. $G_\emptyset(t)$ is the unperturbed generator), and each $G_{(i_1, \dots, i_k)}(t)$ is the partial derivative of $G(t, c_0, \dots, c_{r-1})$ with respect to the variables $c_{i_1}, \dots, c_{i_k}$ (up to the required combinatorial pre-factor for multivariable Taylor series). Hence, they are entirely determined by how $G$ is written in terms of the perturbative parameters $c_0, \dots, c_{r-1}$.

Lastly, we note that the formal algorithms in Section \ref{section:algorithms} solve the above problem in the interaction frame of $G_\emptyset(t)$\cite{evans_timedependent_1967,haeberlen_1968,mehring_1983}, which is reviewed in Section \ref{section:toggling_frame}.

\subsection{Power series notation} \label{section:power_series_notation}

Using Equation \eqref{equation:generator_power_series} as a model, we introduce notation to simplify working with multivariable power series. Each term in Equation \eqref{equation:generator_power_series} is indexed by a list of indices $0 \leq i_1 \leq \dots \leq i_k \leq r-1$. What uniquely identifies the above term is the number of times each index in $\{0, \dots, r-1\}$ appears, and hence a \emph{multi-index} notation \cite{wiki_multi_index_2021} is often used for multivariable power series. We use a slightly different, though functionality equivalent, notation in terms of \emph{multisets}. A multiset is like a set, but in which repeated elements may appear. 

Denoting multisets with round brackets, the multiset associated with the above power series term is given as: $I = (i_1, \dots, i_k)$. Given a list of variables $c_0, \dots, c_{r-1}$, we denote
\begin{equation}
	c_I = c_{i_1} \times \dots \times c_{i_k}.
\end{equation}
For example, $c_{(0, 0, 1)} = c_0^2c_1$, and $c_{(0, 1, 1, 2)} = c_0 c_1^2 c_2$. This notation enables a simple correspondence between algebraic operations and multiset operations. E.g. for two multisets $I, J$, we have that:
\begin{equation}
	 c_{I + J}=c_I \times c_J,
\end{equation}
where $I+J$ denotes the multiset summation.

With this, for $I = (i_1, \dots, i_k)$, the summation term $c_{i_1} \dots c_{i_k} G_{(i_1, \dots, i_k)}(t)$ in Equation \eqref{equation:generator_power_series} is rewritten as
\begin{equation}
	c_I G_I(t).
\end{equation}
Letting $\symI_k(r)$ denote the set of multisets of size $k$ with elements in $\{0, \dots, r-1\}$, and letting $c$ generally denote a list of variables $(c_0, \dots, c_{r-1})$, we may rewrite Equation \eqref{equation:generator_power_series} as
\begin{equation}
	G(t, c) \equiv G(t, c_0, \dots, c_{r-1}) = G_\emptyset(t) + \sum_{k=1}^\infty \sum_{I \in \symI_k(r)} c_I G_I(t). \label{equation:generator_power_series2}
\end{equation}

Lastly, we use $|I|$ to denote the number of elements in the multiset, including repeats. For example, $|(0, 0, 1)| = 3$. This notation is used throughout the paper to represent power series.

\subsection{Interaction frame} \label{section:toggling_frame}

In many quantum control applications, computations are performed in the \emph{interaction frame} of the unperturbed generator \cite{evans_timedependent_1967,haeberlen_1968,mehring_1983}. The benefit of the interaction frame is that it factorizes the evolution generated by $G(t, c)$ into two pieces: one given purely by $G_\emptyset(t)$, and the other being trivial if $c=0$. Denoting
\begin{equation}
	V(t) = \mathcal{T}\exp\left(\int_0^t ds G_\emptyset(s)\right),
\end{equation}
with $\mathcal{T}$ being the \emph{time-ordering operator} \cite{time_ordering_2022}, the generator $G$ in Equation \eqref{equation:generator_power_series2}, transformed into the interaction frame of $G_\emptyset$, is given by:
\begin{equation}
	\tilde{G}(t, c) = \sum_{k=1}^\infty \sum_{I \in \symI_k(r)} c_I \tilde{G}_I(t), \label{equation:toggling_generator}
\end{equation}
where $\tilde{G}_I(t) = V^{-1}(t)G_I(t)V(t)$. With this, it holds that:
\begin{equation}
	\mathcal{T}\exp\left(\int_0^t ds G(s, c)\right) = V(t)\mathcal{T}\exp\left(\int_0^t ds \tilde{G}(s, c)\right). \label{equation:toggling_frame}
\end{equation}

\subsection{Multivariable Dyson series} \label{section:multivariable_dyson}

For a generator $G(t)$, the Dyson series \cite{dyson_radiation_1949} expands 
\begin{equation}
	\mathcal{T}\exp\left(\int_0^t ds G(s)\right) = \id + \sum_{k=1}^\infty D_k(t), \label{equation:standard_dyson}
\end{equation}
where $\id$ is the identity operator, and for the explicit formula
\begin{equation}
	D_k(t) = \int_0^t dt_1 \dots \int_0^{t_{k-1}}dt_k G(t_1) \dots G(t_k). \label{equation:single_dyson}
\end{equation}
We may view this as a power series in a single variable $c$ by using this formula for the generator $c G(t)$:
\begin{equation}
	\mathcal{T}\exp\left(c\int_0^t ds G(s)\right) = \id + \sum_{k=1}^\infty c^kD_k(t),
\end{equation}
where the $D_k(t)$ are still as in Equation \eqref{equation:single_dyson}. 

To define the \emph{multivariable Dyson series}, the goal is to generalize the above equation and explicit expression in Equation \eqref{equation:single_dyson} to an arbitrary number of variables, with the original generator given as a power series.

\begin{definition} \label{definition:multivariable_dyson}
Let $\tilde{G}(t, c)$ be as in Equation \eqref{equation:toggling_generator}, with $c$ representing a list of variables. The \emph{multivariable Dyson series} is the power series of the solution in $c$:
\begin{equation}
	\mathcal{T}\exp\left(\int_0^t ds \tilde{G}(s, c)\right) = \id + \sum_{k=1}^\infty \sum_{I \in \symI_k(r)} c_I \symD_I(t),
\end{equation}
with the $\symD_I(t)$, which we refer to as \emph{multivariable Dyson terms}, defined \emph{implicitly} by the above series decomposition. Note that we also refer to $\symD_I(t)$ simply as \emph{Dyson terms} when there is no risk of ambiguity with the typical ``single-variable'' Dyson series. 
\end{definition}

For a multiset $I$, let $P_k(I)$ denote the set of $k$-fold ordered partitions of $I$, which we take to be empty if $|I| < k$. I.e. $P_k(t)$ denotes the set of ordered lists of proper submultisets $I_1, \dots, I_k \subset I$ such that $I_1 + \dots + I_k = I$. The following Proposition gives an explicit form for the multivariable Dyson terms $\symD_I(t)$.

\begin{proposition} \label{proposition:explicit_dyson}
Letting $\tilde{G}(t, c)$ and the $\symD_I(t)$ be defined as in Definition \ref{definition:multivariable_dyson}, it holds that
\begin{equation} \label{equation:explicit_dyson}
	\symD_I(t) = \sum_{m=1}^{|I|} \sum_{(I_1, \dots, I_m) \in P_m(I)} \int_0^t dt_1 \dots \int_0^{t_{m-1}} dt_m \tilde{G}_{I_1}(t_1) \dots \tilde{G}_{I_m}(t_m)
\end{equation}
\begin{proof}
For the generator $\tilde{G}(s, c)$, the standard Dyson series gives
\begin{equation}
	\mathcal{T}\exp\left(\int_0^t ds \tilde{G}(s, c)\right) = \id + \sum_{k=1}^\infty D_k(t)
\end{equation}
with
\begin{equation} \label{equation:multivariable_single}
	D_k(t) = \int_0^t dt_1 \dots \int_0^{t_{k-1}}dt_k \tilde{G}(t_1, c) \dots \tilde{G}(t_k, c). 
\end{equation}

First, we find an explicit power series decomposition of $D_k(t)$ in the variables $c$. We start by implicitly writing $D_k(t)$ as a power series in the variables $c$:
\begin{equation}
	D_k(t) = \sum_{m=k}^\infty \sum_{I \in I_m(r)} c_I A_{k, I}(t),
\end{equation}
where we start the sum at $m=k$, as all terms below this order will be $0$. To determine the form of the $A_{k,I}(t)$, we expand each instance of $\tilde{G}(t, c)$ in Equation \eqref{equation:multivariable_single} via the power series in Equation \eqref{equation:toggling_generator}, and then collect terms corresponding to each monomial $c_I$. The initial expansion results in a sum of terms of the form:
\begin{equation}
	c_{I_1 + \dots + I_k}\int_0^t dt_1 \dots \int_0^{t_{k-1}}dt_k \tilde{G}_{I_1}(t_1) \dots \tilde{G}_{I_k}(t_k),
\end{equation}
for some list of index multisets $I_j$. Determining the form of $A_{k,I}(t)$ requires identifying all such terms for which $I = I_1 + \dots + I_k$, which corresponds to identifying \emph{ordered partitions} of $I$. Collecting such terms results in the equality
\begin{equation}
	A_{k, I}(t) = \sum_{(I_1, \dots, I_k) \in P_k(I)} \int_0^t dt_1 \dots \int_0^{t_{k-1}}dt_k \tilde{G}_{I_1}(t_1) \dots \tilde{G}_{I_k}(t_k).
\end{equation}
Finally, Equation \eqref{equation:explicit_dyson} is obtained by collecting the $A_{k,I}(t)$ for $k \leq |I|$.
\end{proof}
\end{proposition}

\subsection{Multivariable Magnus expansion} \label{section:multivariable_magnus}

For a generator $G(t)$, the Magnus expansion \cite{magnus_exponential_1954, blanes_magnus_2009} alternatively gives an expansion for an operator:
\begin{equation}
	\Omega(t) = \sum_{k=1}^\infty \Omega_k(t)
\end{equation}
for which, under suitable convergence conditions, gives:
\begin{equation}
	\mathcal{T}\exp\left(\int_0^t ds G(s)\right) = \exp(\Omega(t)). \label{equation:magnus_starting_point}
\end{equation}
Explicit expressions for the $\Omega_k(t)$ can also be given \cite{magnus_exponential_1954, blanes_magnus_2009}; however we do not work with them here. Similarly to the Dyson series, we can consider the Magnus expansion for the generator $cG(t)$ as a single-variable power series expansion. The operator $\Omega(t)$ is often referred to as the \emph{time-averaged generator}. 

\begin{definition} \label{definition:multivariable_magnus}
Let $\tilde{G}(t, c)$ be as in Equation \eqref{equation:toggling_generator}, with $c$ representing a list of variables. The \emph{multivariable Magnus expansion} for the generator $\tilde{G}(t, c)$ is the power series of the time-averaged generator of the solution in $c$:
\begin{equation}
	\Omega(t, c) = \sum_{k=1}^\infty \sum_{I \in \symI_k(r)} c_I \symO_I(t)
\end{equation}
satisfying 
\begin{equation}
	\mathcal{T}\exp\left(\int_0^t ds \tilde{G}(s, c)\right) = \exp(\Omega(t, c)). 
\end{equation}
The $\symO_I(t)$ are defined implicitly according to the above relations, and we refer to them as \emph{multivariable Magnus terms}, or simply as \emph{Magnus terms} when there is no risk of ambiguity with the standard Magnus expansion.
\end{definition}

We do not derive an explicit expression for $\symO_I(t)$ --- we only define them implicitly by the above relation. The algorithm we develop for computing them is based on recursion relations rather than explicit expressions.
\section{Algorithms, scaling, and implementation} \label{section:algorithms}

We now introduce the two algorithms that are the main results of this work: one for computing a collection of multivariable Dyson terms $\symD_I(t)$, and one for computing a collection of multivariable Magnus terms $\symO_I(t)$ recursively from already-computed $\symD_I(t)$. 

\subsection{Computing multivariable Dyson terms} \label{section:dyson}

For a desired collection of $\symD_I(T)$, the algorithm presented here computes all terms simultaneously, along with the solution of the interaction frame propagator $V(T)$, by phrasing them as the solution to a single differential equation, which is then solved via numerical integration. This approach is similar to block-matrix methods used in both the numerical methods \cite{vanloan_1978,carbonell_2008} and quantum literature \cite{haas_engineering_2019,goodwin_2015,machnes_2018}. It may also be viewed as a fully time-dependent generalization of the computation performed in the \emph{pre-computation} step of \code{Dysolve} \cite{shillito_fast_2020}, with similar methods appearing in \cite{kalev_integral-free_2020}. Furthermore, due to the correspondence between power series and derivatives, 
the algorithm we present can alternatively be viewed as a forward-mode sensitivity analysis method \cite{hindmarsh_sundials_2005}, specialized to linear matrix differential equations, that can compute arbitrary order derivatives in a multivariable setting.\footnote{The sensitivity-analysis phrasing implies that these terms can be computed using automatic differentiation applied to differential equation solvers. We have found however that, in practice, recursively calling general automatic differentiation routines to compute higher order derivatives is slow compared to the specialized method presented here. See \cite{rackauckas_comparison_2018} for a discussion of automatic differentiation tools applied to the sensitivity analysis problem for differential equations.}

A technical detail of our algorithm is that, for each index multiset $I$, it computes $\symE_I(T) = V(T)\symD_I(T)$, rather than $\symD_I(T)$, where $V(t)$ is the interaction frame propagator. Depending on application, one may want either $\symE_I(T)$ or $\symD_I(T)$. As $V(T)$ is also computed by the algorithm, this factor can be numerically removed if desired by solving the linear equation 
\begin{equation}
	V(T) X = \symE_I(T)
\end{equation} 
for $X$ ($V(T)$ is always in-principle invertible). When there is no risk of confusion, we will refer to both $\symD_I(T)$, and $\symE_I(T)$ as multivariable Dyson terms.

The algorithm is based on the following proposition, which shows that the derivatives of the $\symE_I(t)$ satisfy a recursion relation. The proof is given in Appendix \ref{appendix:derivations_dyson}.

\begin{proposition} \label{proposition:recursive_derivative}
Let $\tilde{G}(t, c)$ and the $\symD_I(t)$ be as in Definition \ref{definition:multivariable_dyson}. For $\symE_I(t) = V(t)\symD_I(t)$, with $V(t)$ being the interaction frame propagator, it holds that
\begin{equation}
	\dot{\symE}_I(t) = G_\emptyset(t)\symE_I(t) + G_I(t)V(t) + \sum_{J \subsetneq I} G_J(t)\symE_{I \setminus J}(t), \label{equation:dyson_derivative_rule}
\end{equation}
where subsets are understood in terms of multisets. For the case $|I|=1$, the sum over $J \subsetneq I$ is interpreted as being empty (as there are no proper submultisets).
\end{proposition}

Hence, to compute a collection of such terms for a desired list of index multisets $L = \{I_1, \dots, I_j\}$, we need only solve the above differential equation. The formalized algorithm is given in Algorithm \ref{algorithm:dyson_algorithm}. Beyond organizing and structuring the right-hand side function expressed in Equation \eqref{equation:dyson_derivative_rule}, the algorithm must first \emph{complete} the list of index multisets $L = \{I_1, \dots, I_j\}$ in the following sense: Equation \eqref{equation:dyson_derivative_rule} shows that $\dot{\symE}_I(t)$ depends on $\symE_J(t)$ for every $J \subseteq I$. Hence, given a particular set of desired terms $L = \{I_1, \dots, I_j\}$ to compute, constructing a single coupled differential equation to compute them requires first filling out the list $L$ until it is closed under taking subsets. More formally, we say that $L$ is \emph{complete} if for every $I \in L$ and $J \subseteq I$, it holds that $J \in L$. Hence, the first step of Algorithm \ref{algorithm:dyson_algorithm} is finding the completion of the desired terms $L$.

\setcounter{algocf}{0}
\begin{algorithm}[h!]
\caption{Compute multivariable Dyson terms \label{algorithm:dyson_algorithm}}
\KwInput{
\begin{itemize}
	\item Callable matrix-valued function $G_\emptyset(t)$ giving the interaction frame generator.
	\item A list of pairs $(I, G_I(t))$, with $I$ being an index multiset, and $G_I(t)$ a callable matrix-valued function, describing the power series decomposition of the generator. All other power series terms for the generator are assumed by the algorithm to be $0$.
	\item List of index multisets $L = (I_1, \dots, I_j)$ describing the desired Dyson terms to compute.
	\item Integration time $T$.
	\item Boolean flag \code{remove\_V} for whether to return $\symE_I(T)$ or $\symD_I(T)$ for each $I \in L$.
\end{itemize}
}
\KwOutput{
\begin{itemize}
	\item Interaction frame propagator solution $V(T) = \mathcal{T}\exp\left(\int_{0}^{T} dt_1 G_\emptyset(t_1)\right)$.
	\item The completion of $L$, $L'$.
	\item Multivariable Dyson terms (in the frame of $V(t)$), either $\symD_{I}(T)$ if \code{remove\_V == True} or $\symE_{I}(T)$ otherwise, for each $I \in L'$.
\end{itemize}
}
Compute the completion $L'$ of $L$ by recursively looping through each $I \in L$ in order of non-increasing length, adding $I \setminus i$ to $L$ for each $I \in L$

Construct the differential equation:
\begin{substeps}
Canonically order the completed set $L' = I_1, \dots, I_m$

Represent the state of the DE $y(t) = (V(t), \symE_{I_1}(t), \dots, \symE_{I_m}(t))$ with initial condition $y_0 = (I, 0, \dots, 0)$, with $I$ the appropriately-sized identity matrix.

Define RHS function $f(t, y) = (G_\emptyset(t)V(t), G_\emptyset(t)\symE_{I_1}(t) + \sum_{J \subsetneq I_1}G_J(t) \symE_{I_1 \setminus J}(t), \dots)$, treating unspecified $G_J(t)$ as $0$.
\end{substeps}

Solve the differential equation.

\If{\textnormal{\code{remove\_V == True}}}
{
For each $I \in L'$,compute $\symD_I(T)$ by solving $V(T)X = \symE_I(T)$.

Return all the $\symD_I(T)$, and $V(T)$.
}
\Else
{
Return all the $\symE_I(T)$, and $V(T)$.
}
\end{algorithm}

\subsection{Computing multivariable Magnus terms from multivariable Dyson terms} \label{section:magnus_recursion}

Here we develop multivariable generalizations of recursive methods for computing Magnus terms from Dyson terms given in \cite{burum_magnus_1981, salzman_alternative_1985}, though following the notation in \cite[Section 2.4]{blanes_magnus_2009}. Recursive methods provide a compact representation of the computational steps while avoiding the complexity of explicit expressions for Magnus expansion terms \cite{blanes_magnus_2009,arnal_general_2018}.

To review, the recursive methods of \cite{burum_magnus_1981, salzman_alternative_1985} begin by expanding both sides of Equation \eqref{equation:magnus_starting_point}: the left hand is expanded using the Dyson series, and the right-hand side is expanded assuming a series for $\Omega$ and the Taylor series for the exponential. After grouping terms by order and rearranging, \cite{burum_magnus_1981} gives the formula:
\begin{equation}
	\Omega_k = D_k - \sum_{m=2}^k \frac{1}{m!} Q_k^{(m)},
\end{equation}
where
\begin{equation}
	Q_k^{(m)} = \sum_{i_1 + \dots + i_m = k} \Omega_{i_1} \dots \Omega_{i_m}. \label{equation:original_recursion}
\end{equation}
It is then shown in \cite{burum_magnus_1981} that the $Q_k^{(m)}$ matrices satisfy the recursion relation
\begin{equation}
	Q_k^{(m)} = \sum_{j=1}^{k-m+1} Q_j^{(1)} Q_{k-j}^{(m-1)},
\end{equation}
with base case $Q_j^{(1)} = \Omega_j$. 

The following proposition, proven in Appendix \ref{appendix:derivations_magnus}, shows that direct analogues of these formulas hold in the multivariable case.

\begin{proposition} \label{proposition:recursive_magnus}
Let $\symD_I(t)$ and $\symO_I(t)$ be as in Definitions \ref{definition:multivariable_dyson} and \ref{definition:multivariable_magnus}. For all index multisets $I$, it holds that
\begin{equation}
	\symO_I = \symD_I - \sum_{m=2}^{|I|} \frac{1}{m!} \symQ_I^{(m)} \label{equation:sym_magnus_recursion}
\end{equation}
where the matrices $\symQ_I^{(m)}$ are defined as
\begin{equation}
	\symQ_I^{(m)} = \sum_{(I_1, \dots, I_m) \in P_m(I)} \symO_{I_1} \dots \symO_{I_m}, \label{equation:sym_q_def}
\end{equation}
and satisfy the recursion relation
\begin{equation}
	\symQ_I^{(m)} =  \sum_{J \subsetneq I, |J| \leq |I| - (m-1)} \symQ_J^{(1)} \symQ_{I\setminus J}^{(m-1)} \label{equation:sym_q_recursion}
\end{equation}
with base case $\symQ_J^{(1)} = \symO_J$.
\end{proposition}

Hence, our algorithm for computing multivariable Magnus terms is to first compute the corresponding multivariable Dyson terms, then compute the Magnus terms utilizing the above recursion relation. A formal statement of the algorithm, with full bookkeeping details for implementing the recursion relation, is given in Algorithm \ref{algorithm:magnus_from_dyson}.

\begin{algorithm}[h!] 
\caption{Compute multivariable Magnus terms from Dyson terms \label{algorithm:magnus_from_dyson}}
\KwInput{
\begin{itemize}
	\item Complete list of index multisets $L = (I_1, \dots, I_m)$.
	\item Computed Dyson terms $\symD_{I_1}(T), \dots, \symD_{I_m}(T)$.
\end{itemize}
}
\KwOutput{
\begin{itemize}
	\item Magnus terms $\symO_I(T)$ for each $I \in L$.
\end{itemize}
}
Sort $L$ in order of non-decreasing length.

Generate list of ordered pairs $(I, m)$ that index each $\symQ_I^{(m)}$ to compute: 
\begin{substeps}
For each $I \in L$, append $(I, |I|), (I, |I|-1), \dots, (I, 1)$ to the list of pairs.
\tcc{This list is in an order for which the recursion relation for $\symQ_I^{(m)}$ depends only on $\symQ_J^{(n)}$ for which $(J, n)$ appears earlier in the list.}
\end{substeps}

Increment through each pair $(I, m)$ in the list, computing $\symQ_I^{(m)}$:

\If{$|I| = 1$ and $m =1$}
{
Set $\symQ_I^{(1)}(T) = \symD_I(T)$
}
\ElseIf{$|I| > 1$ and $m=1$}
{
Set $\symQ_I^{(1)} = \symD_I(T) - \sum_{m=2}^{|I|} \frac{1}{m!} \symQ_I^{(m)}$
}
\Else
{
Set $\symQ_I^{(m)} =  \sum_{J \subset I, |J| \leq |I| - (m-1)} \symQ_J^{(1)} \symQ_{I\setminus J}^{(m-1)}$
}

\Return{$\symO_I(T) = \symQ_I^{(1)}$ for each $I \in L$.}
\end{algorithm}

\subsection{Scaling} \label{section:scaling}

Here we consider the scaling of the algorithms for computing all terms in the multivariable Dyson series and Magnus expansion in $r$ variables up to truncation order $n$, assuming the user supplies non-zero $G_I(t)$ for all $|I| \leq n$. 

First, the inherent scaling of the problem is determined by the total number of terms at a given truncation order for a given number of variables:
\begin{fact}
The number of terms in a homogeneous multivariate polynomial of order $n$ in $r$ variables is \cite{225963}:
\begin{equation}
	{r + n \choose n} -1. \label{equation:choice_scaling}
\end{equation}
\end{fact}
If either $r$ or $n$ is fixed, the number of terms in the other parameter grows polynomially, which can be seen via the following bound (proof given in Appendix \ref{appendix:scaling_bounds}).
\begin{fact} \label{fact:term_bounds}
It holds that 
\begin{equation}
	{r + n \choose n}  \leq \min(nr^n, rn^r).
\end{equation}
\end{fact}
On the other hand, for large $r$ and $n$, Stirling's approximation gives
\begin{equation}
	{r + n \choose n} \approx \frac{(r+n)^{r+n}}{r^r n^n},
\end{equation}
and thus if both $r$ and $n$ vary, e.g. setting $r=n$, the right hand side yields $2^{2n}$, which is exponential. Hence, asymptotically there is an inherent exponentiality to the problem along the line $r=n$, however for either fixed $r$ or $n$, the scaling in the other parameter is polynomial. 

Finally, for scaling of the algorithms themselves, we have the following bounds on the number of operations:
\begin{fact} \label{fact:algorithm_scaling}
When computing all terms up to order $n$ in $r$ variables:
\begin{itemize}
	\item For computing Dyson terms, evaluating the RHS in Equation \eqref{equation:dyson_derivative_rule} requires
	\begin{equation}
		O\left(\left[{r + n \choose n} -1\right]^2\right)
	\end{equation}
	operations.
	\item Carrying out the recursive procedure outlined in Section \ref{section:magnus_recursion} requires
	\begin{equation}
		O\left(\left[{r + n \choose n} -1\right]^3\right)
	\end{equation}
	operations.
	\end{itemize}
\end{fact}
The proofs of these facts are given in Appendix \ref{appendix:scaling_bounds}. Hence, the algorithms scale polynomially in the number of terms being computed.

We note again this analysis applies to the case of computing \emph{all} terms up to a given order for a given number of variables. In practical applications it may be possible to \emph{a priori} eliminate the need to compute certain subsets of terms, e.g. if they can be shown to be negligible, or, in the case of the Magnus expansion, if terms are known to commute.

\subsection{Implementation} \label{section:implementation}

Algorithms 1 and 2 have been implemented in the function \perturbfunc{} in the \code{perturbation} module of the open-source software package \qiskitdynamics{} \cite{qiskit_dynamics_2021}.\footnote{This function can also be used to compute the Dyson-type terms of \cite{haas_engineering_2019}.} We describe the API here; however note that it may change over time, and hence the package documentation is the best source for up-to-date information.

The goal of the \perturbfunc{} interface is to closely represent the mathematical problems of computing the multivariable Dyson or Magnus terms from the generator power series. Roughly, the inputs are:
\begin{itemize}
	\item A description of the power series for the generator. This description is given by supplying a description of the non-zero $G_I(t)$, in terms of a list of indices given as \code{Multiset} objects \cite{multiset}, and a corresponding list of python-callable functions implementing the $G_I(t)$. Any $G_I(t)$ not explicitly given in this description are assumed to be $0$.
	\item A python callable function implementing the interaction frame operator $G_\emptyset(t)$.
	\item A choice of either Dyson or Magnus expansion.
	\item A description of which terms in the expansion to compute.
	\item Arguments describing how to perform the integration for computing Dyson terms, in the form of the differential equation solver to use, and any optional arguments for the solver.
\end{itemize}
A data storage object is returned providing access to the computed terms.

Additionally, the \code{perturbation} module contains the class \code{ArrayPolynomial} for representing an array-valued polynomial of scalar variables. This class has functionality for both evaluating and manipulating the polynomial. As we will show in the demo in Section \ref{section:control}, this enables direct evaluation of truncated Dyson or Magnus expansions for specific values of the power-series variables, as well as other non-trivial computations involving them. 

Finally, as will be described in more detail in Section \ref{section:numerical_integrators}, numerical integration schemes for linear matrix differential equations based on the Dyson series and Magnus expansion have also been built into Qiskit Dynamics, using the above perturbation functionality. All functionality can be executed using the JAX \cite{jax2018github} array backend. Code examples used to generate the plots in the following section are given in Appendix \ref{app:code}.
\section{Demonstration in robust control problem} \label{section:control}

The goal of robust quantum control is to design control sequences that perform their function in a \emph{region} of model parameter space. A common approach in numerical robust control is to utilize the Dyson series or Magnus expansion to quantify the sensitivity of a control sequence to variations in model parameters. Control sequences are then designed to reduce the impact of terms in these expansions on the overall evolution. In this section we show how the software tools developed here can be used in these workflows. We emphasize that this section is a \emph{software demo}; it shows how the tools can be used in application, and is not itself meant to be new research.

Using a model of a closed-system transmon with a variety of model parameters, we demonstrate how to:
\begin{enumerate}
	\item Compute and use the Magnus expansion to approximate the fidelity of a control sequence, relative to a target gate, in a region of model parameter space. The accuracy of various truncation orders of the Magnus expansion is demonstrated for each model parameter.
	\item Construct a robustness function that utilizes Magnus expansion terms to arbitrary order. Note that we are not advocating for the practical usefulness of this particular construction, but are demonstrating how the tools \emph{enable} such constructions. The usefulness of any such construction must be determined by further research.
\end{enumerate}

The full code for generating the plots can be found in the \code{control\_example.ipynb} Jupyter Notebook in the supplemental repository \cite{supplemental_repo}. In what follows we describe the code as it pertains to the algorithms presented here, but for a full detail see the Notebook.

\subsection{Transmon model and control parameterization}

Here we consider a model of a closed-system transmon, which we model using the Schrodinger equation for unitary evolution:
\begin{equation}
	\dot{U}(t) = -i H(t) U(t),
\end{equation}
where the Hamiltonian $H(t)$ is Hermitian for all $t$. We will phrase things in terms of $H(t)$, rather than $G(t) = - i H(t)$, as in the preceding sections; however note that all previous discussions and formulae can be translated under the symbolic association $G = -iH$.

Let $N$ be the number operator, and $a$ be the annihilation operator, which are infinite dimensional matrices defined as
\begin{equation}
	N = \left(\begin{array}{ccccc}
		0 & 0 & 0 & 0 & \dots \\
		0 & 1 & 0 & 0 &  \dots \\
		0 & 0 & 2 & 0 & \dots \\
		0 & 0 & 0 & 3 & \ddots \\
		\vdots & \vdots & \vdots & \ddots & \ddots
	\end{array}\right) \textnormal{, and }
	a = \left(\begin{array}{ccccc}
		0 & 1 & 0 & 0 & \dots \\
		0 & 0 & \sqrt{2} & 0 &  \dots \\
		0 & 0 & 0 & \sqrt{3} & \dots \\
		0 & 0 & 0 & 0 & \ddots \\
		\vdots & \vdots & \vdots & \ddots & \ddots
	\end{array}\right).
\end{equation}
The model of a single transmon that we use is as follows:
\begin{equation}
	H(t) = 2 \pi \nu N + \pi \alpha N(N - \id) + \frac{\pi}{3}\beta N(N - \id)(N - 2 \id) + s(b, t) 2 \pi r (a + a^\dagger),
\end{equation}
where $\id$ is the identity, $\nu$ is the qubit frequency, $\alpha$ is the anharmonicity, $\beta$ sets the spacing to the $4^{th}$ energy level, $r$ is the drive strength, and finally $s(b, t)$ is the control field, with $b$ representing a vector of parameters describing the control (which we will explicitly parameterize later). For simulation we truncate the transmon to be $5$-dimensional, and choose parameters typical of IBM transmons \cite{malekakhlagh_first-principles_2020}: $\nu = 5.0$, $\alpha = -0.33$, $\beta = -0.015$, and $r = 0.02$, where all times are in ns and frequencies are in GHz.

We expand the model by adding a list of uncertain parameters: (1) perturbations in the frequency $\nu$, (2) perturbations in the anharmonicity $\alpha$, (3) perturbations in the drive strength $r$, (4) addition of a non-linear control term proportional to $s(b, t)^2$, representing non-linearities in the control electronics, (5) perturbations to the higher level spacing given by $\beta$, and (6) perturbations to higher level drive operator elements. These parameters represent common uncertainties in a single transmon model, and we select an assortment to demonstrate the numerical behaviour of our algorithms. Denoting the perturbation parameters according to the above ordering as $c = (c_1, \dots, c_6)$, this modifies the model to:
\begin{equation}
\begin{aligned}
	H(t, c) =& \, 2 \pi \nu ( 1+ c_1) N + \pi \alpha (1 + c_2) N(N - \id) + s(b, t) 2 \pi r (1 + c_3) (a + a^\dagger) \\
	                   & + c_4s(b, t)^2 2 \pi r (a + a^\dagger) + (1 + c_5) \frac{\pi}{3} \beta N(N - \id)(N - 2 \id) \\
	                   & + c_6s(b, t) 2 \pi r P(a + a^\dagger)P,
\end{aligned}
\end{equation}
where $P$ is the orthogonal projection onto the levels above the first two. By collecting terms according to the coefficients in $c$, the relevant structure of the Hamiltonian that we will apply perturbation theory to is:
\begin{equation}
	H(t, c) = H_\emptyset(t) + \sum_{j=1}^6 c_j H_{(j)}(t),
\end{equation}
where $H_\emptyset(t) = H(t, 0)$ is the unperturbed Hamiltonian.

Lastly, we choose a parameterization of the control signal $s(b, t)$. While the exact details aren't particularly important for the demonstration, we choose a parameterization that has desirable properties for an optimization application: it produces smooth and bounded signals that start and end at $0$, and is automatically differentiable with respect to the control parameters. We use the standard representation utilized in Qiskit Dynamics:
\begin{equation}
	s(b, t) = \textnormal{Re}[f(b, t)e^{i 2 \pi \mu t}], \label{equation:general_signal_form}
\end{equation}
where $f(b, t)$ is the parameterized complex-valued envelope, and the signals carrier frequency is $\mu$, the expected frequency of the transmon. Here, we set $\mu = \nu$, the true frequency of the modelled qubit. We construct the envelope f(b, t) as a piecewise constant complex-valued function via the following process. Starting with the input parameters $b$, which are a 2-dimensional real array of shape $(2, k)$, we build and iteratively refine the samples for $f(b,t)$ as follows:
\begin{enumerate}[a)]
	\item First, given a choice of basis vectors $\{d_j\}_{j=1}^k$, construct a pair of \emph{unbounded} sample arrays $x_R = \sum_{j=1}^k b_{0j} d_j$ and $x_I = \sum_{j=1}^k b_{1j} d_j$. That is, $x_R$  and $x_I$ are the starting point for the real and imaginary samples for $f(b, t)$, and the control parameters $b$ are the coefficients in the two linear combinations. Here, we choose $\{d_j\}_{j=1}^k$ to be the discretized Chebyshev basis over the interval $[0, 50]$ ns with sample width $1$ ns, and use $k=8$. Note, these samples are truly unbounded, as we are allowing the entries of $b$ to take \emph{any} real value.
	\item Next, given a diffeomorphism $h : \real \mapsto [-1, 1]$, construct a pair of \emph{bounded} sample arrays as $y_R = h(x_R)$ and $y_I = h(x_I)$ (where we apply $h$ to a vector by applying it independently to all entries). By requiring $h$ to be a diffeomorphism, we ensure that all of the values of $y_R$ and $y_I$ lie in the interval $[-1, 1]$, and that this stage of the construction is automatically differentiable. For this diffeomorphism, we choose $h(x) = \frac{\arctan(x)}{\pi/2}$.
	\item Finally, construct the smoothened complex-valued samples for $f(b, t)$ by convolving $y_R$ and $y_I$ with a smooth kernel. Here, we resample $y_R$ and $y_I$ to be sampled at a rate of $0.125$ ns, then convolve each with a discretized Gaussian function (with an amplitude of $1$ and standard deviation of $0.5$ ns) containing $24$ samples. We normalize the samples of the Gaussian convolution kernel so they sum to $1$, which ensures that the entries of the convolution output lie within, and can achieve all values in, the interval $[-1, 1]$. Denoting $C$ as the resampling and convolution mapping, the final samples for $f(b, t)$ with sample width $0.125$ ns are given by $z = C(y_R) + i C(y_I)$.
\end{enumerate}
The above process is a useful recipe for creating smooth and bounded piecewise constant functions, where the samples are differentiable functions of the unbounded input parameters. See Figure \ref{figure:control_parameterization} for a visualization of the samples constructed at each step of this process. For all of the demonstrations we use random input parameters to the envelope parameterization.

\begin{figure}[h!]
\centering
\begin{subfigure}{.45\linewidth}
    \centering
     \includegraphics[scale=0.5]{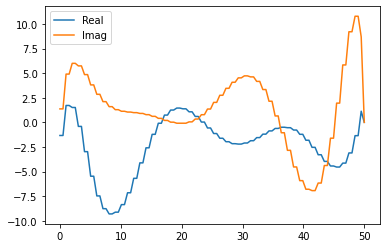}
    \caption{The unbounded real and imaginary sample arrays $x_R = \sum_{j=1}^k b_{0j} d_j$ and $x_I = \sum_{j=1}^k b_{1j} d_j$.}
\end{subfigure}
    \hfill
\begin{subfigure}{.45\linewidth}
    \centering
    \includegraphics[scale=0.5]{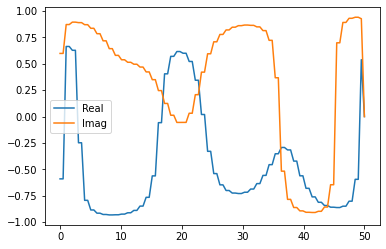}
    \caption{The bounded real and imaginary sample arrays $y_R = h(x_R)$ and $y_I = h(x_I)$, with $h(x) = \frac{\arctan(x)}{\pi / 2}$.}
\end{subfigure}

\bigskip
\begin{subfigure}{\linewidth}
  \centering
  \includegraphics[scale=0.5]{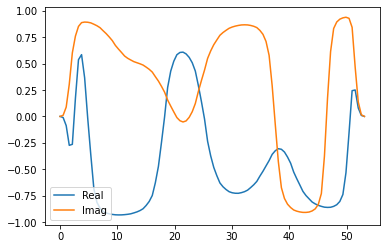}
  \caption{The final convolved real and imaginary sample arrays $C(y_R)$ and $C(y_I)$.}
\end{subfigure} 
\caption{Sample creation pipeline for the piecewise constant envelope $f(b, t)$, for a random choice of input parameters $b$. The sequential transformations (a), (b), and (c), are described in the main text after Equation \eqref{equation:general_signal_form}.}
\label{figure:control_parameterization}
\end{figure}

\subsection{Fidelity approximation via Magnus expansion} \label{section:fidelity_approximation}

Here we use the Magnus expansion to generate approximations of gate infidelity in a region of model parameter space around $c=0$, for a random choice of control parameters $b$.\footnote{Note that we choose to demonstrate Magnus expansion approximation of the infidelity -- rather than the unitary itself -- as the infidelity is typically the most important single metric in numerical control optimization schemes. Furthermore, we consider a random choice of control parameters $b$ (as opposed to a choice close to an implementation of a particular gate), as we are generally interested in how the Magnus expansion captures fidelity in a region of model parameter space, regardless of gate quality. Gate optimization procedures need to navigate through both high and low fidelity regions, and as such it is important to test the approximation quality for randomly chosen control parameters, which are a common starting point for optimization.} Infidelity of a unitary $U$ is defined relative to an $X$ gate on the first two levels of the transmon:
\begin{equation} \label{equation:fidelity_function}
	E(U) = 1-\frac{|Tr(X^\dagger U|_2)|^2}{4},
\end{equation}
where $U|_2$ is the top left $2 \times 2$ block of the unitary $U$. The truncated Magnus expansion, which we denote $\Omega(c)$, as a function of the model parameters $c$ for a fixed $b$, is computed \emph{in the frame of $H_\emptyset(t)$}. The unitary $U(T, c)$ is then approximated as:
\begin{equation}\label{equation:unitary_approx}
	U(T, c) \approx \mathcal{T}\exp\left(-i\int_0^T ds H_\emptyset(s)\right)\exp(\Omega(c)),
\end{equation}
and the approximate infidelity is computed by inputting Equation \eqref{equation:unitary_approx} into the fidelity function in Equation \eqref{equation:fidelity_function}. See Appendix \ref{app:code_solve_lmde_perturbation} for an outline of the code used to compute the above approximation to $U(T, c)$.\footnote{Note that we use the \code{expm} function in JAX to compute the matrix exponential, which utilizes the combination of Pad\'{e} approximation and scaling and squaring given in \cite{al-mohy_new_2009}. Once the Magnus expansion is computed to a desired order, the cost of evaluating Equation \eqref{equation:unitary_approx} for a given $c$ is a function of the number of terms in the expansion, and the matrix dimension. Computing $\Omega(c)$ consists primarily of taking a linear combination of the $d \times d$ matrices in the expansion, which depends linearly on the number of terms in the expansion, and quadratically on the dimension $d$. Computing the matrix exponential doesn't depend on the number of terms, but typically involves many matrix multiplications, which roughly scale as $d^3$, and which we expect to dominate the cost of the computation.}

Figure \ref{figure:perturbation_fidelities} plots the 1d infidelity curve for each perturbation parameter, using the ``true'' unitary $U(T, c)$, as calculated with the JAX \code{odeint} solver with the lowest possible tolerance setting.\footnote{The \code{odeint} solver is based on the Dormand-Prince method outlined in \cite{shampine_practical_1986}, which is a mixed $4^{th}$ and $5^{th}$ order Runge-Kutta variable step-size method.\label{footnote:odeint}} Figure \ref{figure:magnus_fidelity_error} gives 1d plots demonstrating the quality of the above infidelity approximation scheme for various orders of the Magnus expansion. The general trend is that, within a neighbourhood of $0$, the higher order Magnus truncations provide a better approximation. However, as the perturbation parameters become larger, the approximations begin to break down, including the ordering of which order of the Magnus expansion provides a better approximation. Furthermore, some of the infidelity error curves exhibit spurious kinks, as in the Order 1 curve for Perturbation 2. As with any application of perturbation theory, these effects demonstrate that care must always be taken in choosing a truncation order, and in limiting the region of approximation. The $c_5$ perturbation provides an interesting example of numerical breakdown: the perturbation appears to have very little impact on the infidelity, as evidenced by the extremely flat infidelity curve in Figure \ref{figure:perturbation_fidelities}. Over this region, however, we still observe breakdown of the approximations, with higher orders breaking down more quickly. This is plausibly due to the flatness of the infidelity curve: the Magnus expansion terms are near zero, and therefore the computed matrices are dominated by numerical error, which are enhanced at higher order due to the number of matrices involved.

\begin{figure}[h!]
\centering
\includegraphics[scale=0.45]{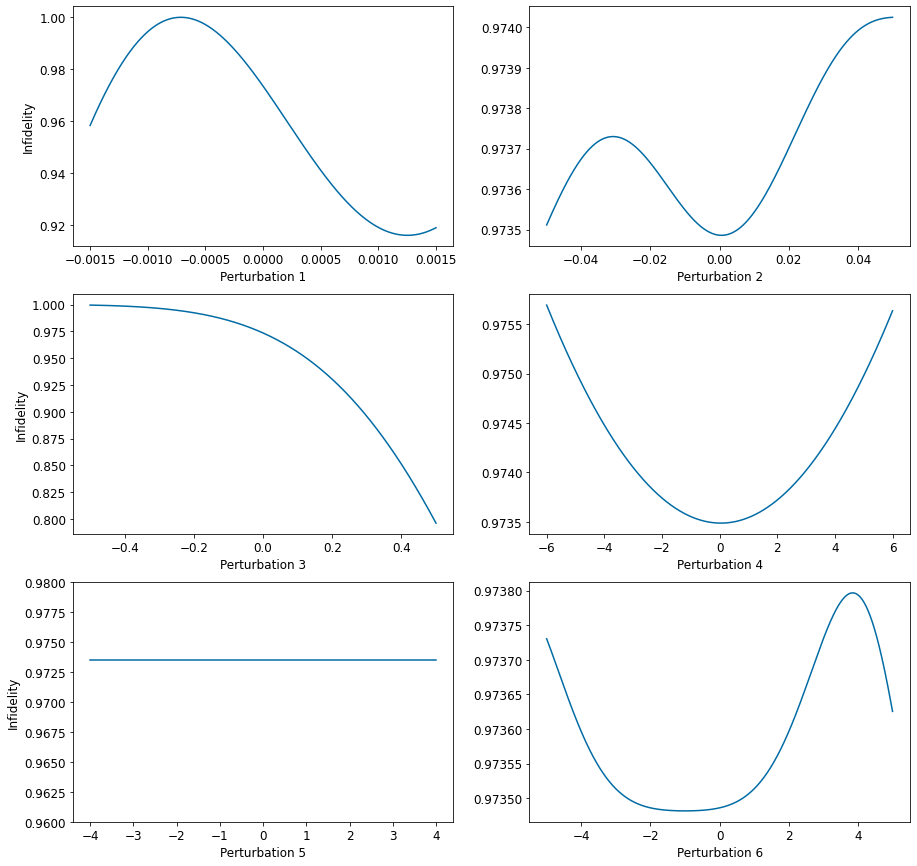}
\caption{Infidelity curves along individual perturbation parameters, holding the others as $0$. For each value of the perturbation parameters $c$ represented by the x-axis in the above plots, the ``true'' unitary $U(T, c)$ is computed using the \code{odeint} solver in JAX with absolute and relative tolerances set to \code{1e-14}, and the infidelity is computed as $E(U(T,c))$. The ranges displayed have been chosen by trial and error based on the impact of the perturbation on the fidelity.}
\label{figure:perturbation_fidelities}
\end{figure}

\def\scalehere{0.25}

\begin{figure}[h!]
\includegraphics[scale=0.45]{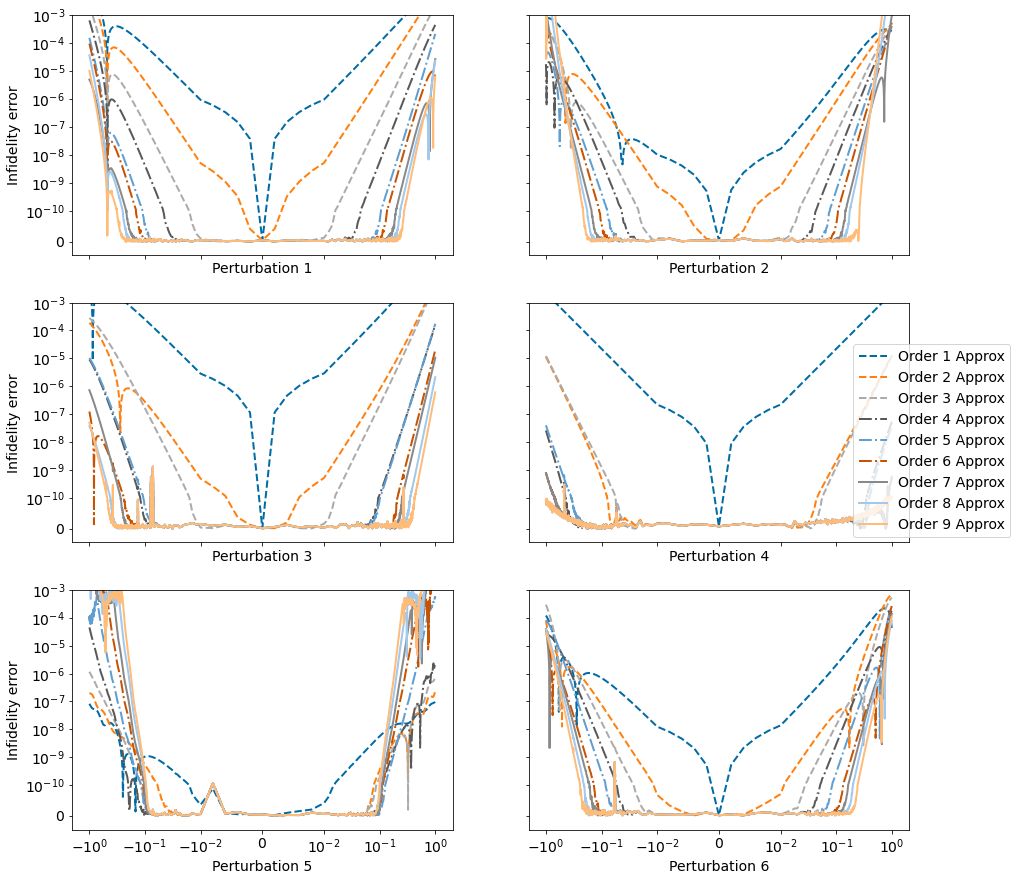}
\caption{Absolute error of the Magnus approximations to infidelity in Figure \ref{figure:perturbation_fidelities} for each perturbation.For a given Magnus order and perturbation value $c$, the approximation $U$ to $U(T, c)$ is computed as in Equation \eqref{equation:unitary_approx}, and the absolute infidelity error is computed $|E(U(T, c)) - E(U)|$. The $y$-axis of each plot is on a logarithmic scale. The $x$-axis for each plot is a percentage of the maximum value displayed for each perturbation in Figure \ref{figure:perturbation_fidelities}, and shown on a logarithmic scale, with linear scaling within $[-10^{-2}, 10^2]$. Negative $x$-axis values correspond to deviations with a minus sign.}
\label{figure:magnus_fidelity_error}
\end{figure}

For perturbations $c_1$ and $c_3$, Figure \ref{figure:magnus_fidelity_error_2d} demonstrates the quality of infidelity approximation given by the Magnus expansion in a 2d plane. For different expansion orders, this plot uses a colormap to distinguish between regions of over and underestimation of the infidelity. With increasing expansion order, we see a growing 2d region of high quality approximation.

\begin{figure}[h!]
\centering
\includegraphics[scale=0.4]{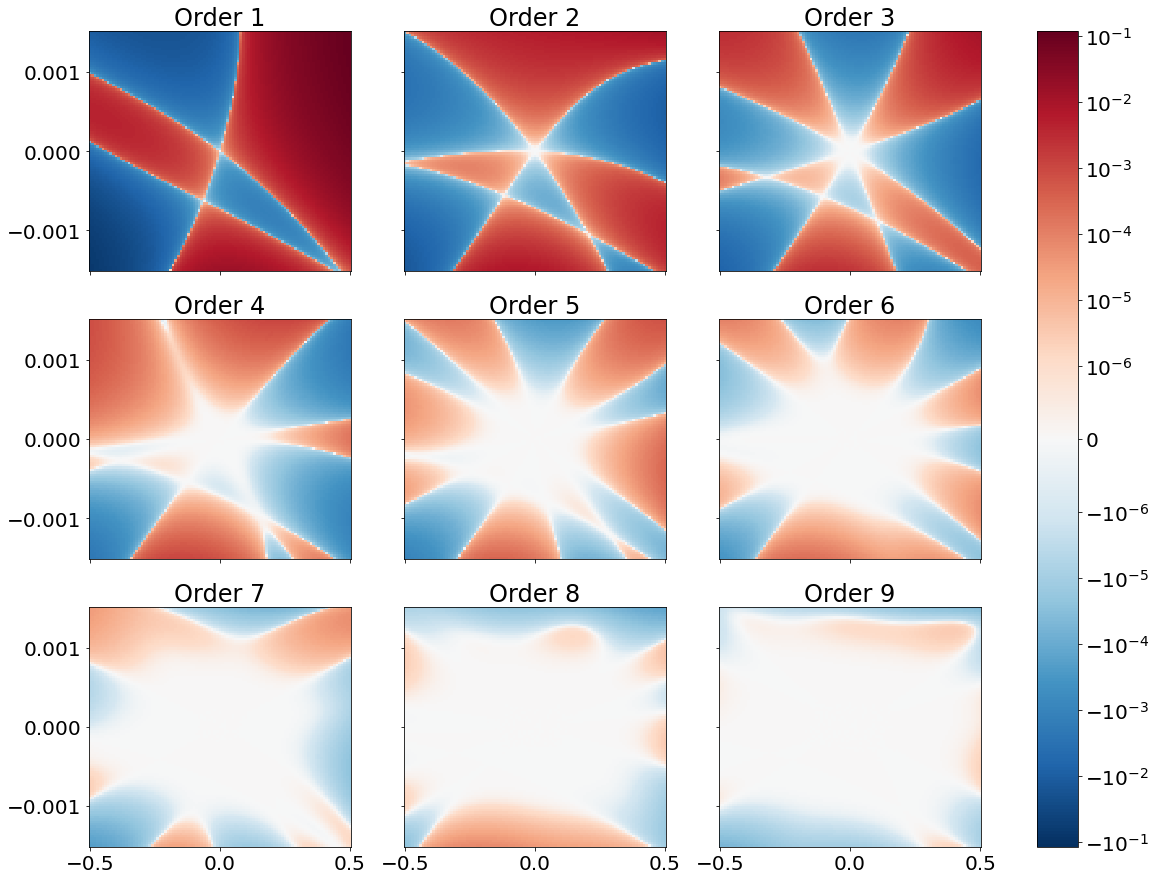}
\caption{For perturbations $c_1$ (y-axis) and $c_3$ (x-axis), the 2d infidelity error, on a symmetric base-10 log scale. Note that in contrast to Figure \ref{figure:magnus_fidelity_error_2d}, for a given perturbation value $c$ and approximation $U$ to the ``true'' unitary $U(T, c)$, these figures plot $E(U) - E(U(T, c))$, without the absolute value. Hence, the plot differentiates regions of both over and underestimation of the infidelity. The scale of the $x$- and $y$-axes are linear, and the colour scale is a symmetric log scale, which is linear over $[-10^{-6}, 10^{-6}]$. The white region roughly indicates where the approximations are achieving an error of magnitude below $10^{-6}$.}
\label{figure:magnus_fidelity_error_2d}
\end{figure}

\subsection{Magnus robustness objective construction} \label{section:robustness}

We now demonstrate the construction of a robustness objective function that aggregates arbitrary order perturbation data. Following the notation in the previous section, but introducing explicit functional dependence on the control parameters $b$, let $\Omega(c, b)$ denote a truncated Magnus expansion in the model parameters $c$, in the interaction frame of $-i \tilde{H}(t, c, b)$. Equation \eqref{equation:unitary_approx} shows that the impact of the perturbations on the evolution is to apply $\exp(\Omega(c, b))$ before the evolution given by $-i \tilde{H}(t, 0, b)$. The robust control literature referenced in the introduction builds robustness objectives based on this fact: if $\Omega(c, b)$ acts trivially on the system, then the perturbations have no effect. Typical approaches consider low order terms and attempt to set them all to $0$. In numerical applications, the goal is to construct an objective function that is minimized if all are $0$. Building these computations is typically done by hand at low orders. We show here how the software tools presented in this paper can be used to extend these approaches to higher orders, for an arbitrary number of perturbations. We emphasize that the practical utility of this approach needs to be investigated: here we simply wish to show how these tools can enable such an investigation.

To build a robustness objective, we need to define a function which is minimized if $\Omega(c, b)$ has no effect on the subspace of interest. In this case, it is the computational subspace given by the first two levels of the transmon model. Furthermore, we need this objective to quantify this over some region of $c$ values. For this, we choose the following:
\begin{equation}
	g(b) = \int_{c \in D} dc p(c) \left\|\Omega(c, b)P - \frac{Tr(\Omega(c, b)P)}{2} P \right\|_2^2,
\end{equation}
where in the above:
\begin{itemize}
	\item $D$ is the region of parameter space of interest, and $p$ is a probability distribution over $D$,
	\item $P$ is the orthogonal projection onto the first two levels, and
	\item $\norm{\cdot}_2$ is the Frobenius norm, defined as $\norm{A}_2 = \sqrt{Tr(A^\dagger A)}$.
\end{itemize}
To understand the integrand, note that, for a matrix $X$, it holds that:
\begin{equation}
	\left\|XP - \frac{Tr(XP)}{2} P \right\|_2^2 = 0
\end{equation}
if and only if $XP = aP$ for some $a \in \complex$, i.e. if $XP$ is proportional to $P$. For quantum systems, this is equivalent to $X$ ``acting trivially'' on the system: such an operator can only impact global phase of the first two levels, and therefore has no impact on the physical evolution restricted to that subspace. Hence, the objective $g$ is $0$ if and only if $\Omega(c, b)$ acts trivially on the computational subspace over a whole region of parameter space.

Effectively choosing the region $D$ and distribution $p$ is itself a challenging problem in multiobjective optimization. This choice effectively sets a linearization strategy for aggregating all Magnus terms of different orders. The distribution impacts which perturbations are favoured in the optimization, and deciding which to weigh more depends on what impact they have on relevant quantities, such as fidelity. Here, we choose $p$ to be a product distribution of Gaussians, i.e.:
\begin{equation}
	p(c) = \prod_{i=1}^r p_i(c_i),
\end{equation}
with $p_i(c_i)$ given by a Gaussian distribution, with the integration region $D$ being rectangular.

To compute $g(b)$, we first identify parts of it that can be pre-computed and reused in each evaluation. Note the following: if $\Omega(c)$ is a multivariable array-valued polynomial in the variables $c$, then so is:
\begin{equation}
	h(c, b) = \left\|\Omega(c, b)P - \frac{Tr(\Omega(c, b)P)}{2} P \right\|_2^2.
\end{equation}
That is, $h(c, b)$ necessarily has a decomposition of the form:
\begin{equation}
	h(c, b) = \sum_{k}\sum_{I \in \symI_k(r)}c_Ih_I(b),
\end{equation}
where we have explicitly included the dependence of the $h_I$ on the control parameters $b$. Assuming such a decomposition (which we will need to compute), we have:
\begin{equation}
	g(b) = \int_{c \in D} dc p(c) h(c, b) = \sum_{k} \sum_{I \in \symI_k(r)} \left(\int_{c \in D} dc p(c) c_I \right) h_I(b).
\end{equation}
Observe that pre-factors $\int_{c \in D} dc p(c) c_I$, which are \emph{moments} of the distribution $p$, are independent of the control parameters $b$, and hence can be pre-computed once and reused in every evaluation of the objective. Hence, given these pre-computed coefficients, we need only compute the $h_I(b)$, and then compute the dot product indicated by the above equation. Note that, under the assumption that both $D$ and $p$ are symmetric under flipping parameter axes, any moments for which the $c_I$ contain odd powers will be $0$, eliminating many terms that need to be computed.

In Appendix \ref{app:robustness_objective_code}, we walk through how Qiskit Dynamics is used to compute the $h_I(b)$, highlighting the correspondence between the mathematical procedure described above and the code. Appendix \ref{app:robustness} shows the scaling of the computation time of both $g(b)$ and its gradient. The scaling is demonstrated with respect to the number of perturbations, for different orders of the Magnus expansion. The CPU plots are generated in the \code{control\_example.ipynb} notebook of the supplemental repository \cite{supplemental_repo}, and the GPU plots in the notebook \code{control\_example\_gpu.ipynb}.
\section{Numerical integrators for systems with fast carrier frequencies} \label{section:numerical_integrators}

The recently introduced \code{Dysolve} algorithm \cite{shillito_fast_2020} utilizes the Dyson series to construct a fixed-step numerical integration method specialized to linear matrix differential equations \emph{with fast carrier frequencies}. By fixing a particular structure of the generator, \code{Dysolve} pre-computes elements of a truncation of the Dyson series, which can then be utilized to repeatedly solve the system for different time-varying envelopes.

In Section \ref{section:perturbative_solvers}, we review the computational problem of \code{Dysolve}, and show how it, and a Magnus-based equivalent, can be phrased in terms of multivariable power series decompositions, and hence the pre-computation step can be performed using the algorithms in Section \ref{section:algorithms}. This leads to a reduction in the number of terms required as reported in \cite{shillito_fast_2020}, in some cases from exponential to polynomial in the truncation order.

In Section \ref{section:integrator_implementation} we describe our implementation of these integrators in the \qiskitdynamics{} package, and in Section \ref{section:solver_demo}, we demonstrate the performance of the implementation, both for solving, as well as for computing gradients with respect to control parameters via automatic differentiation. Finally, in Section \ref{section:perturbative_solver_memory} we discuss the memory requirements of these solvers.

\subsection{Perturbative solvers} \label{section:perturbative_solvers}

The computational problem addressed in \cite{shillito_fast_2020} is to simulate a linear matrix differential equation whose generator is decomposed as
\begin{equation}
	G(t) = F + \sum_{j=1}^s \Re[f_j(t) e^{i \omega_j t}]A_j,
\end{equation}
where $f_j(t)$ are complex-valued envelope functions, $\omega_j$ are carrier frequencies, and $F$ and the $A_j$ are constant matrices. More specifically, the goal is to simulate the above system for different instances of the envelopes $f_j(t)$, while keeping $F$, $A_j$, and the carrier frequencies $\omega_j$ \emph{fixed}. With the assumption that $F$, $A_j$, and $\omega_j$ are fixed, \code{Dysolve} \cite{shillito_fast_2020} utilizes the Dyson series to \emph{pre-compute} aspects of the evolution. Given a fixed time-step $\Delta t$, the approach is to integrate the system over an interval $[t_0, t_0 + \Delta t]$ by computing a truncated Dyson series. While this is generally expensive compared to calling a standard ODE solver, the key observation of \cite{shillito_fast_2020}, phrased in terms of the terminology of this paper, is that the multivariable Dyson series terms associated with a carefully chosen approximate decomposition of the generator are actually independent of both the envelopes and interval start time $t_0$ (up to a frame rotation). Hence, the multivariable expansion terms need only be computed once, and can be reused for arbitrary $t_0$ and envelopes $f_j(t)$.

We now explicitly outline the approach in our notation and terminology, describe a Magnus-based version, and discuss how this framing leads to a reduction in the number of terms required to compute and store. We deviate from \cite{shillito_fast_2020} and consider the generator in the frame of $F$: 
\begin{equation}
	\tilde{G}(t) = \sum_{j=1}^s \Re[f_j(t) e^{i \omega_j t}]\tilde{A}_j(t), \label{equation:numerical_generator}
\end{equation}
where $\tilde{A}_j(t) = e^{-t F} A_j e^{t F}$. Fixing the model details, the method is parameterized in terms of a step size $\Delta t$, and a linear approximation scheme for the envelopes\footnote{Note ``linear'' here refers to the $f_j(t)$ being written as a linear combination of basis functions, \emph{not} that the function is necessarily approximated as linear over each interval.}. I.e. for $t \in [t_0, t_0 + \Delta t]$, the envelopes are approximated as:
\begin{equation}
	f_j(t) \approx \sum_{m = 1}^{d_j} f_{j,m} T_m(t-t_0), \label{equation:env_approx}
\end{equation} 
where the $f_{j,m}$ are the linear approximation coefficients, and the $T_m$ are some chosen basis of functions. The presentation in this section leaves the choice of $T_m$ free, though we assume they are real-valued, and implicitly assume some method of computing the coefficients $f_{j,m}$.\footnote{As described in Section \ref{section:integrator_implementation}, our implementation takes the $T_m(t - t_0)$ to be the Chebyshev polynomials defined on the interval $[t_0, t_0 + \Delta t]$, and the $f_{j,m}$ are computed via Discrete Chebyshev Transformation.}

Using the envelope approximations in Equation \eqref{equation:env_approx} for $t \in [t_0, t_0 + \Delta t]$, the generator $\tilde{G}(t)$ is approximated as:
\begin{equation}
\begin{aligned}
	\tilde{G}(t) &\approx \sum_{j=1}^s \sum_{m=0}^{d_j} \Re\left[f_{j,m}e^{i \omega_j t_0}\right]\cos(\omega_j (t-t_0))T_m(t-t_0)\tilde{A}_j(t) \\
	&\qquad + \sum_{j=1}^s \sum_{m=0}^{d_j} \Im\left[f_{j,m}e^{i \omega_j t_0}\right]\sin(-\omega_j (t-t_0))T_m(t-t_0)\tilde{A}_j(t), \label{equation:approx_generator}
\end{aligned}
\end{equation}
where we have used that
\begin{equation}\label{equation:real_imag_decomposition}
\begin{aligned}
	\Re[f_{j,m}T_m(t-t_0)e^{i \omega_j t}] = \Re&\left[f_{j,m}e^{i \omega_j t_0}\right]\cos(\omega_j (t-t_0))T_m(t-t_0)\\ &+ \Im\left[f_{j,m}e^{i \omega_j t_0}\right]\sin(-\omega_j (t-t_0))T_m(t-t_0).
\end{aligned}
\end{equation}

The decomposition in Equation \eqref{equation:approx_generator} gives the generator as approximately equal to a linear combination of time-dependent terms, with the property that dependence of the generator on the envelope functions is entirely compartmentalized in the constant coefficients. It is with respect to this decomposition that we compute either the multivariable Dyson series or Magnus expansion over the interval $[t_0, t_0 + \Delta t]$: 
\begin{itemize}
	\item The coefficients of the expansion $c_{(0)}, \dots, c_{(r-1)}$ are the 
	\begin{equation}
		\Re\left[a_{j,m}e^{i \omega_j t_0}\right]\textnormal{ and }\Im\left[a_{j,m}e^{i \omega_j t_0}\right]
	\end{equation}
	in some prescribed order, and
	\item The corresponding time-dependent operators $G_{(0)}(t_0, t), \dots, G_{(r-1)}(t_0, t)$ are given by
	\begin{equation}
	\cos(\omega_j (t-t_0))T_m(t-t_0) A_j \label{equation:cosine_operator}
	\end{equation}
	and
\begin{equation}
\sin(-\omega_j (t-t_0))T_m(t-t_0) A_j \label{equation:sine_operator},
\end{equation}
where we have explicitly included the start time $t_0$ in the signature of the $G_j(t_0, t)$.
\end{itemize}
In the above, $r = 2\sum_{j=1}^s d_j$, corresponding to the number of terms appearing in the approximate generator in Equation \eqref{equation:approx_generator}.

Denote $\symD_I(t_0, t_0 + \Delta t)$ and $\symO_I(t_0, t_0 + \Delta t)$ as the multivariable Dyson and Magnus terms for a given index multiset $I$ for the above operators, in the frame $F$, with integration time $[t_0, t_0 + \Delta t]$. Critically, these operators satisfy the following time-translation identities:
\begin{equation}
	\symD_I(t_0, t_0 + \Delta t) = e^{-t_0F}\symD_I(0, \Delta t)e^{t_0 F},\label{equation:dyson_time_translation}
\end{equation}
and:
\begin{equation}
	\symO_I(t_0, t_0 + \Delta t) = e^{-t_0F}\symO_I(0, \Delta t)e^{t_0 F}. \label{equation:magnus_time_translation}
\end{equation}
That is, the power series terms over an interval of length $\Delta t$ can be translated to different start times via a frame transformation. This was shown in \cite{shillito_fast_2020} for the Dyson series, and we provide a proof for both the Dyson series and Magnus expansion for our particular setup in Appendix \ref{app:time_translation}.

Thus, after pre-computing the relevant collection of $\symD_I(0, \Delta t)$ or $\symO_I(0, \Delta t)$, the solution $U(t_0, t_0 + \Delta t)$ over the interval $[t_0, t_0 + \Delta t]$ for the generator $\tilde{G}(t)$ is approximated via the following steps:
\begin{itemize}
	\item Compute the series variables $c_i$ by computing the envelope approximation coefficients $f_{j,m}$, and computing $\Re\left[f_{j,m}e^{i \omega_j t_0}\right]$ and $\Im\left[f_{j,m}e^{i \omega_j t_0}\right]$ in suitable order.
	\item Evaluate the truncated series $\sum_I c_I \symD_I(0, \Delta t)$ or $\sum_I c_I \symO_I(0, \Delta t)$.
	\item In the case of the Magnus expansion, exponentiate the above results.
	\item Conjugate the results from the previous step by $e^{-t_0F} (\cdot) e^{t_0 F}$ to translate the truncated series to the right starting time $t_0$.
\end{itemize}
In Appendix \ref{app:numerical_integration_algorithm} we describe how to save on the frame translation steps when simulating over a contiguous series of intervals.

We end by noting the differences between our presentation and that of \cite{shillito_fast_2020} in the case of the Dyson series. First, consider the number of perturbation terms required to compute, store, and take linear combinations of. According to Section \ref{section:scaling}, utilizing an $n^{th}$ order in the series, given the approximate generator decomposition of Equation \eqref{equation:approx_generator}, requires ${n+r \choose n} - 1$, with $r = 2 \sum_{j=1}^s d_j$. Ref. \cite{shillito_fast_2020} states that when approximating to the $n^{th}$ order in the Dyson series, for a single term in the sum in Equation \eqref{equation:numerical_generator} (corresponding to $s=1$), and using only a single term in the envelope decomposition in Equation \eqref{equation:env_approx} (corresponding to $d_1 = 1$), the \code{Dysolve} algorithm requires computing $2^{n+1} - 1$ perturbation terms. However, in this case, we have $r=2$, and hence our version only requires
\begin{equation}
	{n+2 \choose n}  - 1 = \frac{(n+2)(n+1)}{2} - 1
\end{equation}
terms for expansion order $n$, which is polynomial rather than exponential in the order.

Second, in this paper we set up the problem in the frame of the constant operator $F$, then truncate the Dyson series at a given order, whereas \cite{shillito_fast_2020} expands the propagator in the ``lab frame''. In terms of the resulting operators, the two are related by left-multiplication by $e^{-\Delta t F}$, so there is no fundamental difference between the two choices.

\subsection{Implementation in Qiskit Dynamics} \label{section:integrator_implementation}

The classes \code{DysonSolver} and \code{MagnusSolver} implement the perturbative solvers described above. For the envelope approximation over each interval, we choose to use Chebyshev polynomials, with the approximation coefficients in Equation \eqref{equation:env_approx} computed via Discrete Chebyshev Transform. The matrix exponentials in \code{MagnusSolver} are computed using the \code{expm} function in JAX, which uses the method in \cite{al-mohy_new_2009}. These solvers have been implemented to evaluate the approximate solution over all sub-intervals simultaneously in a vectorized way, and the overall solution is then computed by multiplying these together.

The API and behaviour of these classes, as of version $0.3.0$, are described in Appendix \ref{app:perturbative_solver_api}. The main goal of this implementation is to enable easy configuration of the various solver parameters, to enable performance investigations as in the next section.

\subsection{Performance demonstration} \label{section:solver_demo}

The potential performance gains of these methods are demonstrated in \cite{shillito_fast_2020}, in which their implementation of the method is compared against the solver in QuTiP \cite{johansson_qutip_2013}. Here, we compare the performance of our implementation against the traditional ODE solvers available in Qiskit Dynamics.\footnote{These comparisons were performed using the Qiskit Dynamics \code{main} branch on commit \code{948809} (to be included in version 0.4.1), JAX version 0.4.8, and CUDA 12.}

We consider a model of two interacting transmons in the Duffing approximation described by Hamiltonian:
\begin{equation}
\begin{aligned}
	H(t) =&\; 2 \pi \nu_0 N_0 + \pi \alpha_0 N_0(N_0 - \id_0)  \\
		&+ 2 \pi \nu_1 N_1 + \pi \alpha_1 N_0(N_1 - \id_1) \\
		& + 2 \pi J (a_0 a_1^\dagger + a_0^\dagger a_1) \\
		& + s_0(t) \times 2 \pi (a_0 + a_0^\dagger) \\
		& + s_1(t) \times 2 \pi(a_1 + a_1^\dagger),
\end{aligned}
\end{equation}
where for transmon $j \in \{0, 1\}$:
\begin{itemize}
	\item $a_j$, $N_j$, and $\id_j$ are the raising, number, and identity operators,
	\item $\nu_j$ is the frequency and $\alpha_j$ is the anharmonicity,
	\item $J$ is the coupling strength, and
	\item $s_j(t)=Re[f_j(t)e^{i\omega_jt}]$ is the drive signal on the transmon.
\end{itemize}
The parameters are chosen to model qubit pair $[{\rm Q}3, {\rm Q}5]$ of \code{ibmq\_montreal} as reported on 04/19/2022:\\
\begin{itemize}
	\item $\nu_0 = 5.105$, $\nu_1 = 5.033$, $\alpha_0 = -0.33516$, $\alpha_1 = -0.33721$, and $J = 0.002$,
\end{itemize}
where Q3 is indexed by $0$, and Q5 by $1$.

For the drive signals $s_0(t)$ and $s_1(t)$, we choose the Direct CX pulse described in Fig. 7 of \cite{Jurcevic_2021}, with the control being qubit $0$ (modeling Montreal Q3), and the target being qubit $1$ (modeling Montreal Q5). The Direct CX is a maximally entangling $ZX$ rotation between the two qubits facilitated by three drive tones at the target qubit's transition frequency $\nu_1$. One drive tone is applied to the control qubit, and the other two pulses are symmetric and asymmetric pulses of variable amplitude and phase applied to the target qubit. The shapes are shown in Figure \ref{figure:signal_shapes}, and the explicit mathematical forms, as functions of the amplitude $A$, gate time $T$, risetime $r$ (the length of time over which the pulse goes from $0$ to max amplitude), and $\sigma$ (width of shape during risetime), are given in Appendix \ref{app:pulse_form}. In the simulations below, we treat the amplitudes and phases of the three pulses as variable parameters, while $T = 200$, $r=7$, and $\sigma = 7$, are fixed. All times and frequencies are in ns and GHz.

\def\scalehere{0.5}

\begin{figure}
\centering
\includegraphics[scale=\scalehere]{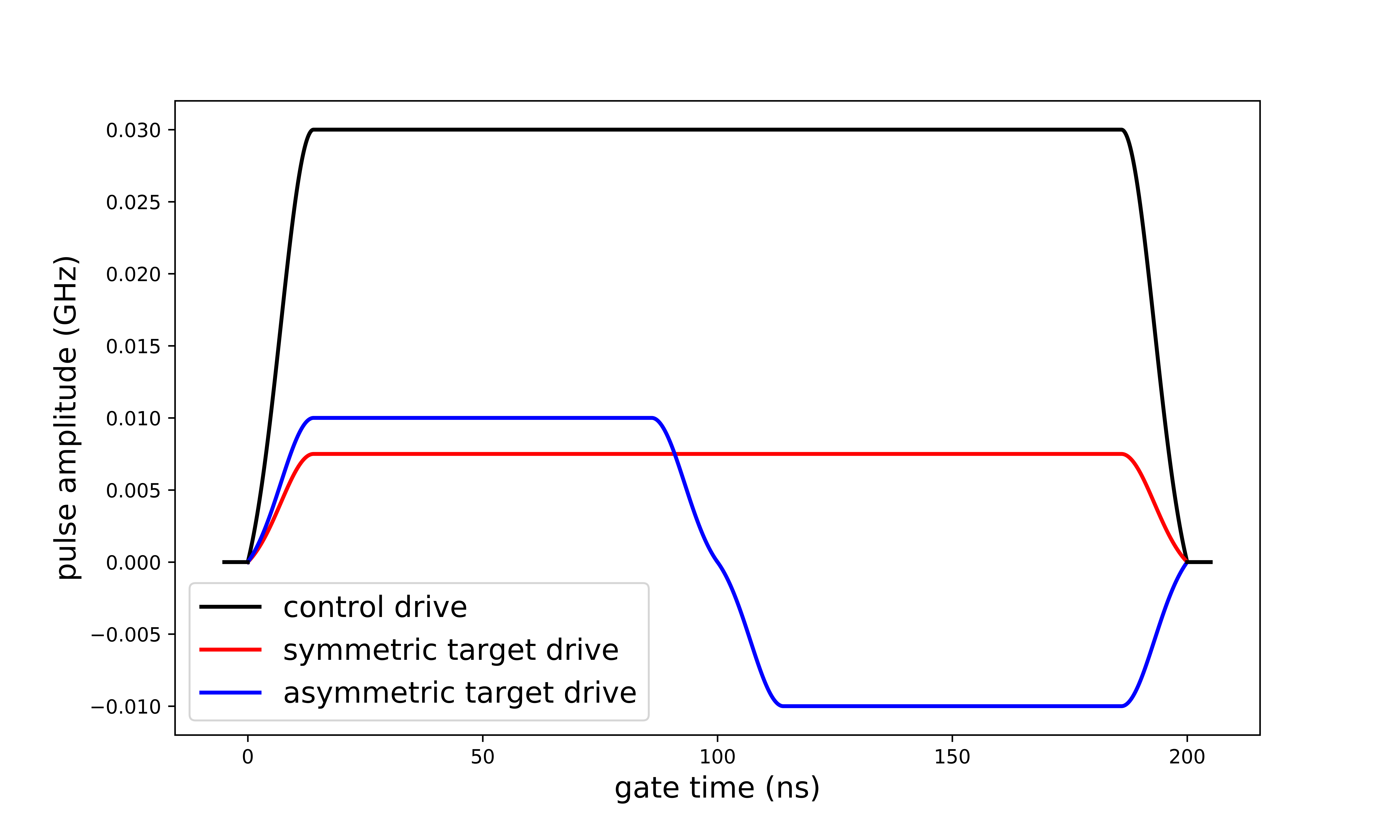} 
\caption{Direct CX qubit drive envelopes. The ``control drive'' envelope is the envelope of $s_0(t)$, applied to qubit $0$. Both the ``symmetric target drive'' and ``asymmetric target drive'' envelopes are added together to form the envelope of $s_1(t)$, applied to qubit $1$, which is the ``target'' for the CX gate. The carrier frequency of both signals are set to the frequency of qubit $1$.\label{figure:signal_shapes}}
\end{figure}

The \code{benchmarking} folder in the supplementary code repository \cite{supplemental_repo} contains scripts for comparing the Dyson and Magnus perturbative solvers against standard ODE solvers for this problem.\footnote{The exact data used in the plots presented in this section is stored in the \code{benchmarking/data} folder.} We consider both the time to generate the solution, as well as the time to compute the gradient of some real-valued function of the solution with respect to the pulse parameters. For setting up the perturbative solvers, we set the static part and frame operator to:
\begin{equation}
\begin{aligned}
	F = -i(&2 \pi \nu_0 N_0 + \pi \alpha_0 N_0(N_0 - \id_0)  \\
		&+ 2 \pi \nu_1 N_1 + \pi \alpha_1 N_0(N_1 - \id_1) \\
		& + 2 \pi J (a_0 a_1^\dagger + a_0^\dagger a_1)),
\end{aligned}
\end{equation}
and the control operators to:
\begin{equation}
	A_0 = -i 2 \pi (a_0 + a_0^\dagger)\textnormal{ and } A_1 = -i 2 \pi (a_1 + a_1^\dagger).
\end{equation}
Note that $A_0$ corresponds to the operator acting on the control qubit, and hence is modulated by $s_0(t)$ whose envelope is given by the control drive in Figure \ref{figure:signal_shapes}, and $A_1$ is the operator acting on the target qubit, which is modulated by $s_1(t)$ whose envelope is the sum of the target drives in Figure \ref{figure:signal_shapes}. For setting up standard ODE solvers, we use the \code{Solver} class in Qiskit Dynamics, specifying the Hamiltonian decomposition required by it, and also specifying the simulation to be in the rotating frame of the static part (we have found that this generally improves performance). The \code{odeint} ODE solver available in JAX \cite{jax2018github} is used (see footnote \ref{footnote:odeint} for a description of \code{odeint}). For all solvers, we use a truncation dimension of $5$ for each transmon, for a total system dimension of $25$.

To gain a picture of the trade-off between speed and accuracy of both \code{odeint} and the perturbative solvers, we choose many different configurations of each solver, then compare, for randomly chosen input control parameters, the total time to solve vs the average accuracy of the generated solutions relative to benchmark solutions.\footnote{The same random input parameters are used for all configurations, to ensure a fair comparison.} This approach is chosen as:
\begin{itemize}
	\item It allows for comparison of solvers with very different forms of configuration (e.g. tolerances vs fixed step size), and
	\item We do not know a priori what is the best way to choose all of the parameters of the perturbative solvers, and want to get an empirical sense of what speed vs accuracy ratios are possible.
\end{itemize}
Additionally, to fairly compare the perturbative solvers to \code{odeint}, it is necessary to account for the fact that the perturbative solvers naturally utilize parallelization.\footnote{The perturbative solvers consist of operations that can be naturally parallelized, whereas the ODE solvers are fundamentally serial computations. As such, comparisons running a single simulation at a time would be misrepresentative of performance; a single run of a traditional ODE solver may take more time due to the number of serial steps, but actually use far less of a device at any given time (either a GPU, or a multicore CPU). This is especially true for the system dimension considered here, which is small.} As such, on GPU we use the JAX \code{vmap} transformation to \emph{vectorize} calls to \code{odeint}, effectively parallelizing calls to this solver. As shown in Figure \ref{figure:gpu_parallelization_saturation} in Appendix \ref{app:speed_comparison}, the benefits of this parallelization for \code{odeint} on GPU are saturated around $7000$ inputs, so we choose this number of inputs for speed comparisons to the perturbative solvers. By contrast there appear to be no benefits to vectorizing our implementation of the perturbative solvers on GPU (which already consist of parallel operations), and hence for these solvers we run the $7000$ simulations using a serial loop. For CPU comparisons, which are shown in Appendix \ref{app:speed_comparison}, we use a single core for both \code{odeint} and the perturbative solvers, as we have found that the speed of both types of solvers scales sub-linearly with the number of cores utilized. The CPU comparisons are performed with only $100$ inputs, which is enough to average out performance variations and just-in-time compilation costs for a single core.

For measuring solution accuracy, we use the distance metric between solutions $U, V$:
\begin{equation}
	\norm{U - V}_2 / \sqrt{d},
\end{equation}
where $d$ is the dimension of the whole space. This metric is chosen as it is used for the error estimation in the \code{odeint} solver. Using this metric, the accuracy of a given solution at the final time is measured via the distance to a benchmark solution, computed using \code{odeint} with tolerances set to \code{1e-14} (using 64-bit precision).

In terms of configurations, for \code{odeint}, we consider various choices of tolerances, setting both \code{rtol} and \code{atol} to all powers of $10$ within the range of \code{1e-6} to \code{1e-14}. For both the Dyson and Magnus perturbative solvers, we consider all possible combinations of the following parameter choices:
\begin{itemize}
	\item The Chebyshev approximation order of each signal, either $0$, $1$, or $2$,
	\item The order of the expansion, from $2$ to $5$.
	\item The number of time steps used to simulate over $T=200$. We divide the interval into $M=10^4$, $2 \times 10^4$, $3 \times 10^4$, $4 \times 10^4$, or $5 \times 10^4$, and set $dt = T / M$.  
\end{itemize}

Figure \ref{figure:solver_data_gpu} shows the speed v.s. accuracy trade-off for the best performing solver configurations when run on the Nvidia A100 GPU with 80GB of memory. The Dyson and Magnus-based solvers outperform the traditional ODE solver to various degrees depending on the accuracy level. As shown in Figure \ref{figure:speedup_gpu}, the speedups from using the perturbative solvers range from $2 \times$ to $4 \times$, and for gradient computations, from about $10 \times$ to $60 \times$. The full data, showing performance on GPU for all tested solver configurations, is shown in Figure \ref{figure:solver_data_gpu_full} in Appendix \ref{app:speed_comparison}.

\begin{figure}[h!]
\centering
\hspace*{-1.25cm}
\includegraphics[scale=.4]{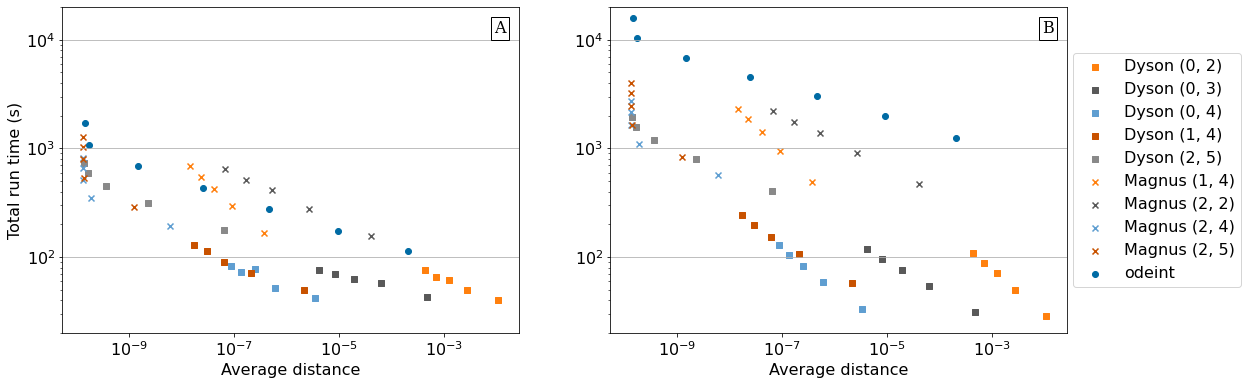}
\caption{Total Runtime v.s. Average Distance for (A) computing the final unitary for 7000 randomized control parameters and (B) computing the gradient of the fidelity of the final unitary with respect to the same pulse parameters, computed on an Nvidia A100 GPU with 80GB of memory. For both Dyson and Magnus solvers, the label $(m, n)$ denotes the configuration parameters used: a $m$-order Chebyshev approximation of the signals, and an $n^{th}$ order perturbative expansion. For a given solution $U$, distance is computed via the metric $\norm{U - V}_2 / \sqrt{d}$, where $V$ is a benchmark solution computed using \code{odeint} at absolute and relative tolerances \code{1e-14}. Average distance is the arithmetic mean of these values. For each configuration, the data points correspond to various numbers of time-steps $M=10^4$, $2 \times 10^4$, $3 \times 10^4$, $4 \times 10^4$, or $5 \times 10^4$. This plot only contains select data of the best performing perturbative solver configurations. The \code{odeint} solutions were computed by vectorizing the computation over all inputs at once, and the perturbative solvers computed all solutions in a serial loop. Full data for both CPU and GPU-based simulations can be found in Appendix \ref{app:speed_comparison}. Note the peculiar feature that some points in plot (B) lie below the corresponding points in plot (A), indicating that the gradient computation is actually faster in some instances. While the fidelity computation of plot (B) seemingly involves more computation than the unitary computation in plot (A), we believe this is due to optimizations in the JAX compilation taking advantage of the fidelity only depending on a sub-block of the full unitary. \label{figure:solver_data_gpu}}
\end{figure}

\begin{figure}[h!]
\centering
\hspace*{-1.25cm}
\includegraphics[scale=.45]{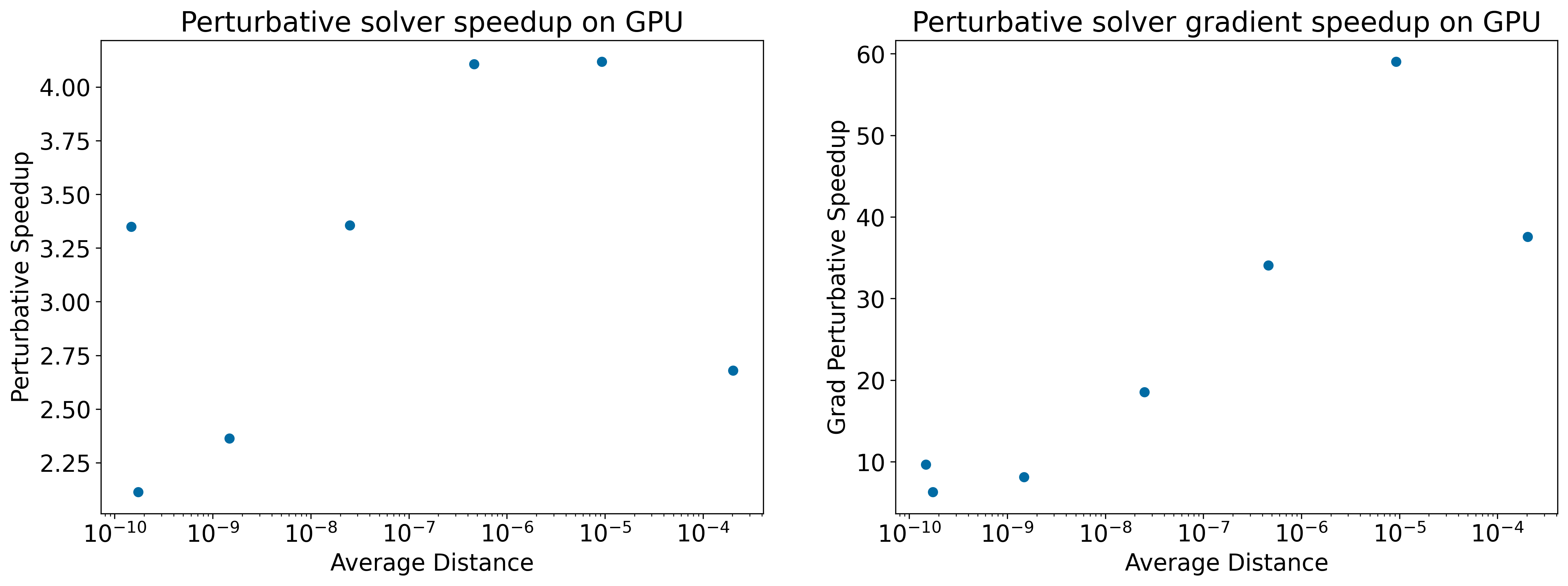}
\caption{Speedup vs average distance for data in Figure \ref{figure:solver_data_gpu}. For each ODE solver point, the speedup is computing the solution time ratio of the fastest perturbative solver (drawing from both Dyson and Magnus solvers) achieving an average distance no worse than the ODE solver point. Plot (A) shows the speedup for computing the final unitary $U$, in which we observe a speedup of roughly $2\times$ at very low average distances, and up to $4\times$ for higher average distances. Plot (B) demonstrates the corresponding data for computing solution gradients. For gradients, the speedup is more dramatic, potentially owing to the perturbative solvers being an ``easier'' computation to automatically differentiate as compared to traditional ODE solvers.}
\label{figure:speedup_gpu}
\end{figure}

The same speed comparisons are shown for CPU in Figures \ref{figure:solver_data_cpu_full} and \ref{figure:solver_data_cpu_grad_full} in Appendix \ref{app:speed_comparison}. For the CPU case, we conclude that this implementation of the perturbative solvers is unlikely to be practically useful. At best, modest speedups are observed, and due to the memory requirements of the perturbative solvers, they will not easily scale up to be run in parallel on many CPU cores. See Appendix \ref{app:speed_comparison} for details. As the speedups for these solvers are likely very problem-dependent, more real-world usage is necessary to draw more conclusions. They may get better or worse as the dimension of the system scales, and will certainly be limited by the number of time-dependent signals appearing in the generator decomposition in Equation \eqref{equation:numerical_generator}.

Lastly, the full data plots for GPU in Appendix \ref{app:speed_comparison} show that the Magnus solver tends to be slower than the Dyson solver, however it becomes comparable, and even slightly faster than the Dyson solver for the highest accuracy solutions. Furthermore, the Magnus solver generally seems to produce higher accuracy solutions: on average, the Magnus solver points appear to have lower average distance. We suspect that these observations are the result of a trade-off with the Magnus-based approach: it requires a matrix exponential at every step, which is costly, but requires fewer expansion terms to achieve the same level of accuracy as the Dyson-based approach. This latter point is evidenced by Figure \ref{figure:terms_vs_accuracy}, which plots the average distance attained by the Dyson and Magnus solvers vs the number of expansion terms used in a given configuration. Generally, for the same number of terms and time-steps, the Magnus solver achieves better accuracy, and hence may be of more use in memory-limited scenarios.

\begin{figure}
\centering
\hspace*{-1.125cm}
\includegraphics[scale=.55]{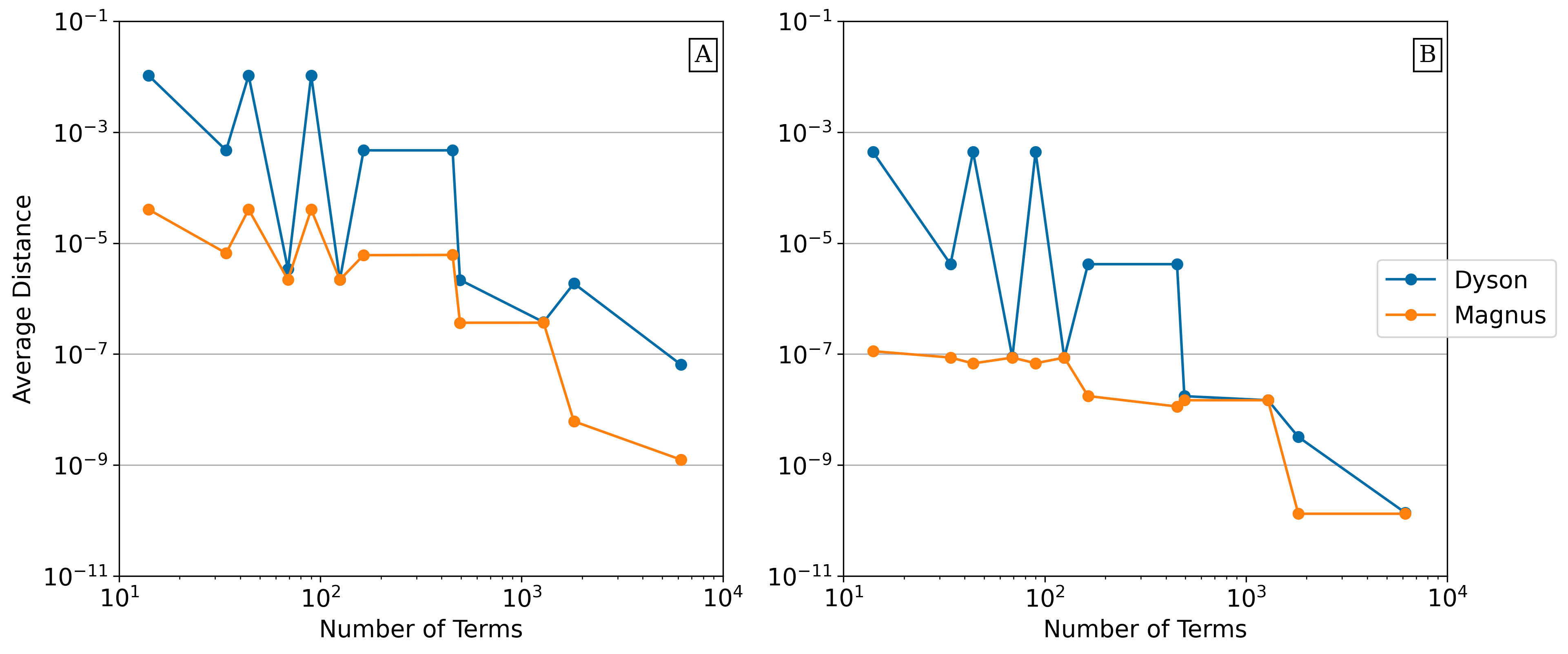}
\caption{The average distance achieved vs the number of terms used in the expansion for both Dyson and Magnus for (A) $M=10^4$ time-steps and (B) $M=5 \times 10^4$ time-steps. The number of terms is determined as follows: Section \ref{section:scaling} shows that, for $r$ variables and expansion order $n$, there are ${r + n \choose n} -1$ terms to compute and store in the expansion. In this case, the number of variables is $r = 2$(number of signals)$\times 2$(real/imaginary parts)$\times$\code{chebyshev\_order}. For a given number of terms, the Magnus solver always achieves a lower average distance per number of terms than the Dyson solver, in some cases significantly so. This data supports the notion that the Magnus solvers may be of more use in memory-limited scenarios, in which fewer terms can be stored and utilized. \label{figure:terms_vs_accuracy}}
\end{figure}

\subsection{Memory requirements and potential improvements} \label{section:perturbative_solver_memory}

Lastly, we end this section with a discussion on memory requirements of the perturbative solvers. To simplify the discussion, assume that all signal envelopes are approximately decomposed as a linear combination of $d_c$ functions in Equation \eqref{equation:env_approx}. As described in Section \ref{section:perturbative_solvers}, the number of terms appearing in the approximate generator at each time step are $r = 2 s d_c$, where $s$ is the number of signals. Hence, for a Dyson or Magnus expansion of order $n$, the pre-computation step for the perturbative solvers involves solving a differential equation whose solution is an array of
\begin{equation}
	{n+2 sd_c \choose n} - 1 \label{equation:combinatorial_expression}
\end{equation}
complex matrices of dimension $d$. For the example in the preceding section, where the matrix dimension is $25$ and we use $64$ bit floating point numbers, each matrix requires $10$ kilobytes of memory. For the case $d_c = 3$ (corresponds to Chebyshev order $2$) and $n = 5$, Equation \eqref{equation:combinatorial_expression} evaluates to $6187$, and therefore the storage of the perturbative expansion terms requires about $59$ megabytes of memory. The pre-computation step itself therefore should require some small constant multiple of $59$ megabytes of memory, depending on the number of copies of the state the ODE solver stores at any given time. For these parameters, we observe roughly $700$ megabytes of memory usage on a personal laptop during the pre-computation step, which is about $10$ to $12\times$ the memory required to simply store the matrices. 

After the pre-computation step is completed, to solve over $M$ time steps, we have implemented the solvers to evaluate the approximate solution over each time step simultaneously using vectorization. The solution at the final time is then computed by multiplying these together using the \code{associative\_scan} looping construct in JAX, which utilizes vectorized operations to simultaneously carry out independent computations. For this step it is difficult to exactly estimate the memory requirements, as the underlying compilation can merge operations in non-trivial ways, and would require a detailed understanding of the algorithm underlying the \code{associative\_scan} function. For example, for $M=40000$, storing the $M$ matrices (with $d=25$) requires about $381$ megabytes of memory. However, when running the Dyson solver for $M=40000$ and the parameters in the preceding paragraph, it uses about $10$ gigabytes of memory. 

As we have observed, even with the heavy memory usage of this implementation, the solvers provide an advantage over traditional solvers on GPU. Tradeoffs between memory usage and speed of the solvers could be further explored by changing how parallelization is utilized. For example, the memory usage could be decreased using a serial loop to construct and propagate the solution over each time step, rather than using \code{associative\_scan}. While this would slow down an individual instance of the solver, it would become more feasible to run multiple instances in parallel. Alternatively, a custom version of \code{associative\_scan} could be written to only compute the solution at the final time (or a list of desired times), eliminating some redundancy in its output (the definition of the problem \code{associative\_scan} solves, for our problem, amounts to explicitly computing the solution at $t_0 + k \Delta t$ for all $k$ up to $M$).

Lastly, the heavy memory requirements of these solvers impact the system size they can feasibly be applied to. Even at dimension $25$ they use significant memory, and as such, applying them to larger systems will require tweaking the methods to reduce memory footprint. Whether or not they provide an advantage in larger systems is an open question warranting further investigation.
\section{Discussion} \label{section:discussion}

We have developed algorithms and software tools for numerically working with the Dyson series and Magnus expansion in a multivariable setting, with the goal of facilitating numerical research applications in quantum control and device engineering. We have demonstrated the tools in the context of robust control, and built higher-level numerical tools in the form of solvers based on the Dyson series and Magnus expansion, demonstrating speed ups on GPU in a simulation of a two transmon gate.

In terms of the Dyson and Magnus-based solvers, given their speed on GPU, it would be helpful to develop more automated ways of choosing the parameters, or to devise a variable-step version of them. Given the phrasing of the pre-computation step in this paper as an ODE, the step size could be dynamically adjusted by solving the ODE from the current step size to a new one. The comparison of low order vs higher order formula typically utilized in variable-step solvers could be done between different expansion orders and signal envelope approximation orders.

Finally, an important observation is that the ODE derived to compute Dyson terms in Algorithm \ref{algorithm:dyson_algorithm} for linear matrix differential equations \emph{is itself a linear matrix differential equation}. That is, the right hand side is a linear function of the state, and as it only consists left-multiplication of the state, it can be rewritten as a large matrix multiplying the state (e.g. as is done in \cite{haas_engineering_2019}). As such, the Dyson and Magnus-solvers could be applied to this ODE, which could speed up the computations performed in the robust control demo when executed on GPU.

\subsection*{Acknowledgements}

We thank Holger Haas and Haggai Landa for helpful discussions and feedback. We are also grateful to the reviewers for improving the clarity of the paper with their careful reading and helpful comments.

\appendix
\section{Proofs} \label{appendix:derivations}

\subsection{Multivariable Dyson term recursion relation} \label{appendix:derivations_dyson}

\begin{proof}[Proof of Proposition \ref{proposition:recursive_derivative}]
For an index multiset $I$ with $|I| \geq 2$, we can relate $\symD_I(t)$ recursively to lower order terms utilizing the explicit expression given in Proposition \ref{proposition:explicit_dyson}:
\begin{equation}
\begin{aligned}
	\symD_I(t) &= \sum_{m=1}^{|I|} \sum_{(I_1, \dots, I_m) \in P_m(I)} \int_0^t dt_1 \dots \int_0^{t_{m-1}}dt_m \tilde{G}_{I_1}(t_1) \dots \tilde{G}_{I_m}(t_m) \\
	&= \int_0^t dt_1 \tilde{G}_I(t) + \sum_{m=2}^{|I|} \sum_{(I_1, \dots, I_m) \in P_m(I)} \int_0^t dt_1 \dots \int_0^{t_{m-1}}dt_m \tilde{G}_{I_1}(t_1) \dots \tilde{G}_{I_m}(t_m).
\end{aligned}
\end{equation}
We can reorganize the sum: sum first over $J \subsetneq I$ for the first partition in the integral, and then sum over all ordered partitions of $I \setminus J$: 
\begin{equation}
\begin{aligned}
	\symD_I(t) &= \int_0^t dt_1 \tilde{G}_I(t)\\ &\quad + \sum_{J \subsetneq I} \sum_{m=1}^{|I \setminus J|} \sum_{(I_1, \dots, I_m) \in P_m(I \setminus J)} \int_0^t ds \tilde{G}_J(s) \int_0^s dt_1 \dots \int_0^{t_{m-1}}dt_m \tilde{G}_{I_1}(t_1) \dots \tilde{G}_{I_m}(t_m) \\
	&= \int_0^t dt_1 \tilde{G}_I(t) + \sum_{J \subsetneq I} \int_0^t dt_1 \tilde{G}_J(t_1) \symD_{I \setminus J}(t_1).
\end{aligned}
\end{equation}
Left-multiplying the above equation by $V(t)$ on both sides and differentiating yields Equation \eqref{equation:dyson_derivative_rule}.
\end{proof}

\subsection{Multivariable Magnus terms recursion relation} \label{appendix:derivations_magnus}

\begin{proof}[Proof of Proposition \ref{proposition:recursive_magnus}]
Here we prove Equations \eqref{equation:sym_magnus_recursion} and \eqref{equation:sym_q_recursion} by walking through the multivariable analogue of the derivations in \cite{burum_magnus_1981, salzman_alternative_1985}. Starting from the definitions of the multivariable Dyson series and Magnus expansion, we have:
\begin{equation}
	\id + \sum_{k=1}^\infty \sum_{I \in \symI_k(r)} c_I \symD_I(t) = \exp\left(\sum_{k=1}^\infty \sum_{I \in \symI_k(r)} c_I \symO_I(t)\right).
\end{equation}
Expanding the right hand side using the Taylor series for the exponential gives:
\begin{equation}
	\id + \sum_{k=1}^\infty \sum_{I \in \symI_k(r)} c_I \symD_I(t) = \id + \sum_{m=1}^\infty \frac{1}{m!}\left( \sum_{k=1}^\infty \sum_{I \in \symI_k(r)} c_I \symO_I(t)\right)^m.
\end{equation}
By expanding out the power, and recollecting terms in the coefficients $c_I$ we may observe that
\begin{equation}
	\left( \sum_{k=1}^\infty \sum_{I \in \symI_k(r)} c_I \symO_I(t)\right)^m = \sum_{k=m}^\infty \sum_{I \in \symI_k(r)} c_I \sum_{(I_1, \dots, I_m) \in P_m(I)} \symO_{I_1}(t) \dots \symO_{I_m}(t),
\end{equation}
where $k$ starts at $m$ as the $m^{th}$ power only contains elements $c_I$ for which $|I| \geq m$.
This leads to the power series equality
\begin{equation}
	\sum_{k=1}^\infty \sum_{I \in \symI_k(r)} c_I \symD_I(t) = \sum_{m=1}^\infty \frac{1}{m!}\sum_{k=m}^\infty \sum_{I \in \symI_k(r)} c_I \sum_{(I_1, \dots, I_m) \in P_m(I)} \symO_{I_1}(t) \dots \symO_{I_m}(t).
\end{equation}
Swapping the order of summation for $k$ and $m$ on the right hand side yields:
\begin{equation}
	\sum_{k=1}^\infty \sum_{I \in \symI_k(r)} c_I \symD_I(t) = \sum_{k=1}^\infty \sum_{I \in \symI_k(r)} c_I\left(\symO_I(t) + \sum_{m=2}^{|I|} \frac{1}{m!}\sum_{(I_1, \dots, I_m) \in P_m(I)} \symO_{I_1}(t) \dots \symO_{I_m}(t)\right).
\end{equation}
Equating terms with the same coefficients leads to the recursion relation:
\begin{equation}
	\symO_I(t) = \symD_I(t) - \sum_{m=2}^{|I|} \frac{1}{m!} \sum_{(I_1, \dots, I_m) \in P_m(I)} \symO_{I_1}(t) \dots \symO_{I_m}(t). \label{equation:sym_magnus_recursion_app}
\end{equation}

Defining $\symQ_I^{(m)}$ as in Equation \eqref{equation:sym_q_def}:
\begin{equation}
	\symQ_I^{(m)} = \sum_{(I_1, \dots, I_m) \in P_m(I)} \symO_{I_1} \dots \symO_{I_m}
\end{equation}
and subbing into Equation \eqref{equation:sym_magnus_recursion_app} yields Equation \eqref{equation:sym_magnus_recursion}.

To derive the recursion relation in Equation \eqref{equation:sym_q_recursion}, we use the same trick as in the Dyson case: break the sum over the first multiset $I_1$ and a sum over the remaining members of the ordered partition:
\begin{equation}
\begin{aligned}
	\symQ_I^{(m)} &= \sum_{J \subset I, |J| \leq |I| - (m-1)} \symO_J \sum_{(I_1, \dots, I_{m-1}) \in P_{m-1}(I\setminus J)} \symO_{I_1} \dots \symO_{I_{m-1}} \\ 
	     & = \sum_{J \subset I, |J| \leq |I| - (m-1)} \symO_J \symQ_{I\setminus J}^{(m-1)} \\
	     &=  \sum_{J \subset I, |J| \leq |I| - (m-1)} \symQ_J^{(1)} \symQ_{I\setminus J}^{(m-1)}.
\end{aligned}
\end{equation}
\end{proof}

\subsection{Algorithmic scaling bounds} \label{appendix:scaling_bounds}

\begin{proof}[Proof of Fact \ref{fact:term_bounds}]
It holds that
\begin{equation}
	{r + n \choose n} -1 = \sum_{m=1}^n\sum_{1 \leq i_1 \leq \dots \leq i_m \leq r} 1 \leq \sum_{m=1}^n \sum_{1 \leq i_1, \dots, i_m \leq r} 1 = 1 + r + \dots + r^n \leq 1 + nr^n,
\end{equation}
where in the inequality we have removed the ordering requirements on $i_1, \dots, i_m$. This establishes the required upper bound $nr^n$, and the upper bound $rn^r$ follows from ${r + n \choose n}$ being symmetric under the exchange $r \leftrightarrow n$.
\end{proof}

\begin{proof}[Proof of Fact \ref{fact:algorithm_scaling}]
For computing the RHS for all terms in Equation \eqref{equation:dyson_derivative_rule}, observe that computing the derivative of a single $\symE_I(t)$ requires computing as many products as there are submultisets of $I$. The number of submultisets is trivially bounded by the total number of terms being computed. Summing this bound over all terms being computed results in the computation of the RHS requiring the square of the number of terms of matrix multiplications and matrix additions, yielding the desired bound.

The recursive algorithm for computing Magnus terms from Dyson terms consists of computing all of the $\symQ_I^{(m)}$ matrices. For a fixed index multiset $I$, there are $|I|$ values of $m$. The worst case cost of computing a given $\symQ_I^{(m)}$ occurs when $m > 1$. Computing a single such term involves performing as many matrix multiplications and additions as there are submultisets $J \subset I$ satisfying $|J| \leq |I| - (m-1)$. For a coarse bound, we can simply bound the number of such sets with the total number of terms being computed, ${r + n \choose n} - 1$. Hence, computing a single term $\symQ_I^{(m)}$ requires
\begin{equation}
	O\left({r + n \choose n} - 1\right)
\end{equation}
operations. Finally, note that there are ${r + n \choose n} - 1$ values of $I$, and for each $I$, $|I|$ values of $m$. We may coarsely bound $|I| \leq n \leq {r + n \choose n} - 1$, and hence the cost of performing the recursion relation for \emph{all} terms requires 
\begin{equation}
	O\left(\left[{r + n \choose n} -1\right]^3\right)
\end{equation}
operations.
\end{proof}
\section{Robustness objective computation time} \label{app:robustness}

The robustness objective computation time of Section \ref{section:robustness} is given in Figure \ref{figure:magnus_objective_times} for CPU, and Figure \ref{figure:magnus_objective_times_gpu} for GPU.

\begin{figure}[h!]
\begin{tabular}{cc}
\includegraphics[scale=\scalehere]{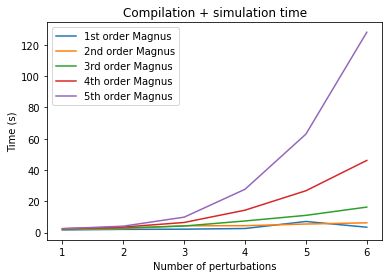} & \includegraphics[scale=\scalehere]{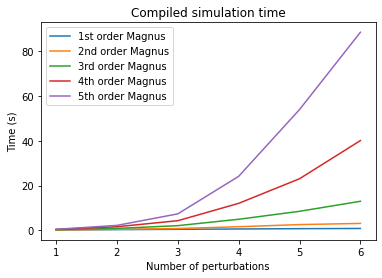}\\ 
\includegraphics[scale=\scalehere]{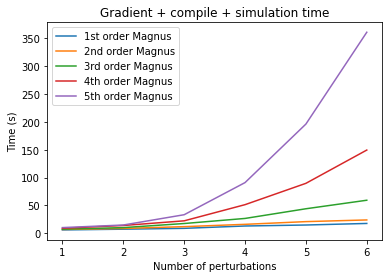} & \includegraphics[scale=\scalehere]{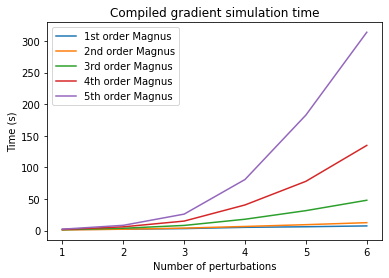}
\end{tabular}
\caption{CPU timing plots for computing the robustness $g(b)$ in Section \ref{section:robustness}. The top left plot shows the cost of computing $g(b)$, including the JAX compilation time, and the top right shows the cost once compiled. Similarly, the bottom left plot shows the cost of computing the gradient of $g$, including compilation, and the bottom right shows the cost of the gradient once compiled. All plots are generated on the CPU of a personal laptop.}
\label{figure:magnus_objective_times}
\end{figure}

\def\scaleheregpu{0.25}

\begin{figure}[h!]
\begin{tabular}{cc}
\includegraphics[scale=\scaleheregpu]{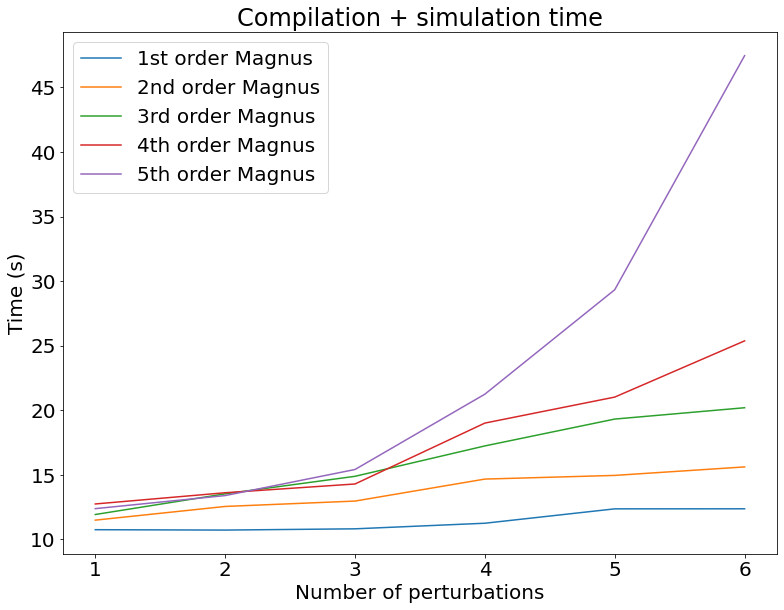} & \includegraphics[scale=\scaleheregpu]{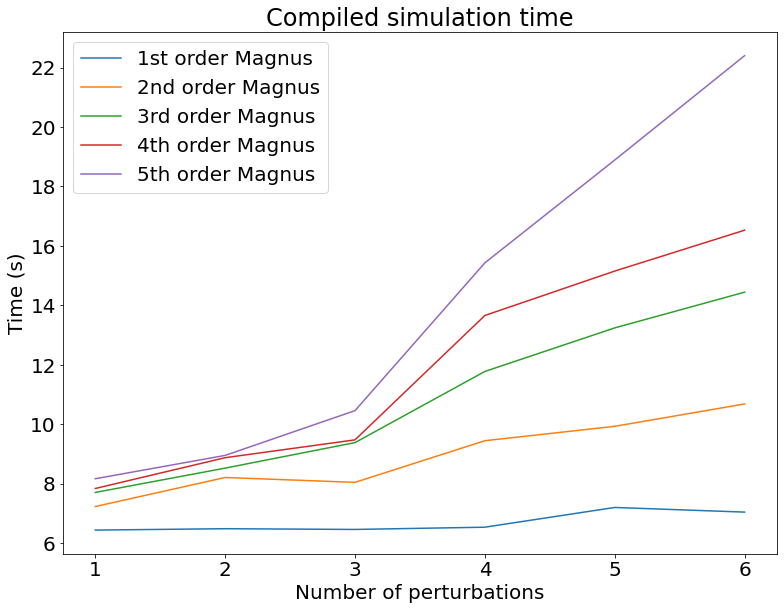}\\
\includegraphics[scale=\scaleheregpu]{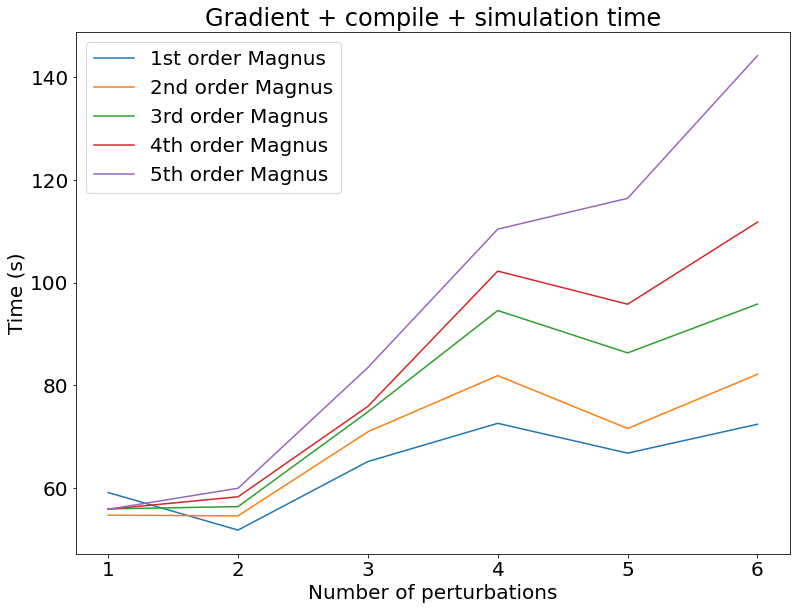} & \includegraphics[scale=\scaleheregpu]{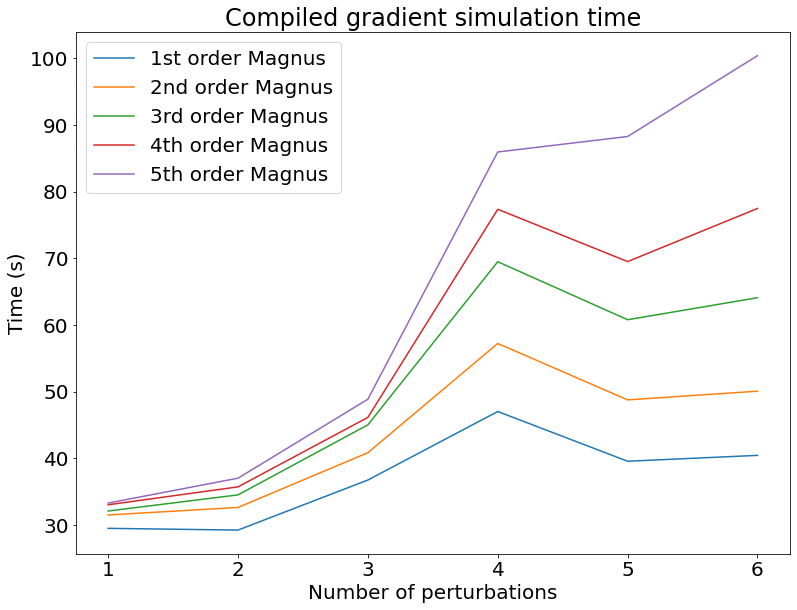}
\end{tabular}
\caption{The GPU equivalent of Figure \ref{figure:magnus_objective_times}. Computations were done using an Nvidia A100 GPU with 80GB of memory.}
\label{figure:magnus_objective_times_gpu}
\end{figure}

\clearpage
\section{Numerical integrators for systems with fast carrier frequencies}
\label{app:numerical_integrators}

\subsection{Time translation formula derivation} \label{app:time_translation}

Starting from Equation \eqref{equation:approx_generator}, we can rewrite:
\begin{equation}
	\tilde{A}_j(t) = e^{-tF}A_je^{tF} = e^{-t_0F}\left(e^{-(t-t_0)F}A_je^{(t-t_0)F}\right) e^{t_0F} = e^{-t_0F}\tilde{A}_j(t-t_0) e^{t_0F}.
\end{equation}
Subbing this into Equations \eqref{equation:cosine_operator} and \eqref{equation:sine_operator}, we find that the $G_I(t_0, t)$ operators satisfy
\begin{equation}
	\tilde{G}(t_0, t) = e^{-tF}G_I(t_0, t)e^{tF} = e^{-t_0F}\tilde{G}_I(0, t-t_0)e^{t_0F}. \label{equation:generator_time_translation}
\end{equation}
Finally, note that the $\symD_I(t_0, t_0 + \Delta t)$ and $\symO_I(t_0, t_0 + \Delta t)$ are both linear combinations of integrals of the form (or, of products of integrals of the following form):
\begin{equation}
	\int_{t_0}^{t_0 + \Delta t} dt_1  \dots \int_{t_0}^{t_{k-1}} dt_k \tilde{G}(t_0, t_{\sigma_1}) \dots \tilde{G}(t_0, t_{\sigma_k})
\end{equation} 
for some permutation $\sigma$. Equation \eqref{equation:generator_time_translation} gives that the above integral is equal to 
\begin{equation}
	e^{-t_0F}\left(\int_{t_0}^{t_0 + \Delta t} dt_1  \dots \int_{t_0}^{t_{k-1}} dt_k \tilde{G}(0, t_{\sigma_1} - t_0) \dots \tilde{G}(0, t_{\sigma_k} - t_0)\right)e^{t_0F},
\end{equation} 
and performing a change of variables $s_j = t_j - t_0$ for all integration variables yields
\begin{equation}
	e^{-t_0F}\left(\int_{0}^{\Delta t} ds_1  \dots \int_{0}^{s_{k-1}} ds_k \tilde{G}(0, s_{\sigma_1}) \dots \tilde{G}(0, s_{\sigma_k})\right)e^{t_0F}.
\end{equation} 
Equations \eqref{equation:dyson_time_translation} and \eqref{equation:magnus_time_translation} immediately follow.

\subsection{Simulation frame simplification} \label{app:numerical_integration_algorithm}

When simulating over a contiguous series of intervals with endpoints $t_k = t_0 + k \Delta t$, the frame handling steps can be simplified. Let $W_k$ denote the approximation to the propagator at the $k^{th}$ interval, \emph{before} performing the final frame translation step. Over $M$ contiguous intervals starting at $t_0$, the approximate evolution in the frame $F$ is then given as:
\begin{equation}
	e^{-t_{M-1}F}W_{M-1}e^{t_{M-1}F} \dots e^{-t_0F}W_0e^{t_0F}.
\end{equation}
Each intermediate frame product is of the form: $e^{t_{k+1}F}e^{-t_kF} = e^{\Delta t F}$. Hence, rather than computing $e^{t_k F}$ for every $t_k$, we only need to compute 3 exponentials $e^{-t_{M}F}$, $e^{t_0F}$, and $e^{\Delta t F}$, and each time-step requires only one extra multiplication as opposed to two. 

In the case of the Magnus expansion no further simplifications are possible, but for the Dyson series, the pre-factor $e^{\Delta t F}$ can be included directly in the decomposition as part of the pre-computation step. This corresponds to using the $\symE_I(t)$ version of the symmetric Dyson operators, and fully eliminates the additional multiplications for frame handling at each step.

\subsection{Functional form of Direct CX pulses} \label{app:pulse_form}

Denoting $A$ the pulse amplitude, $T$ the total gate time, $r$ the risetime, and $\sigma$ the width of the shape during the risetime, the envelope functions for the Direct CX gate are as follows. First, denote the functions:
\begin{equation}
	C(r) = e^{-\frac{r^2}{2}}
\end{equation}
\begin{equation}
	D(r, \sigma) = \sigma\left(\sqrt{2 \pi} \erf\left(\frac{r}{\sqrt{2}}\right) - 2r C(r)\right).
\end{equation}
The symmetric target and control drive envelopes are parameterized by
\begin{equation}
	f(t, A, \sigma, r, T) = \begin{cases} 
      A\left(e^{-\frac{1}{2}\left(\frac{t - r \sigma}{\sigma}\right)^2} - C(r)\right)/D(r, \sigma)   & t < r\sigma \\
      A(1-C(r)) / D(r, \sigma) & r\sigma \leq t \leq T - r\sigma \\
      A\left(e^{-\frac{1}{2}\left(\frac{(T-t) - r \sigma}{\sigma}\right)^2} - C(r)\right)/D(r, \sigma) & T - r\sigma \leq t
   \end{cases},
\end{equation}
and the anti-symmetric target drive is parameterized by
\begin{equation}
	g(t, A, \sigma, r, T) =  \begin{cases} 
      f(t, A, \sigma, r, T/2)   & t < T/2 \\
      -f(t - T/2, A, \sigma, r, T/2) & T/2 \leq t
      \end{cases}.
\end{equation}

\subsection{Speed comparison supplement} \label{app:speed_comparison}

The plot showing the benefits of parallelization of \code{odeint} on GPU via vectorization is shown in Figure \ref{figure:gpu_parallelization_saturation}. The ratio of the total simulation time over the number of simulations, when all simulations are run simultaneously via vectorization, is plotted as a function of the number of simulations. The curve decreases steadily until it becomes nearly flat over $6000$ simulations. This indicates that the time to simulate scales sub-linearly until about $6000$ simulations, and then scales roughly linearly beyond this point. We interpret this as a sign that parallelization on GPU using \code{odeint} for our problem is beneficial until $6000$, and beyond this there is no additional speed gain provided by vectorization. This motivates the choice of $7000$ inputs to use for speed comparisons on GPU.

\begin{figure}[h!]
\centering
\includegraphics[scale=0.6]{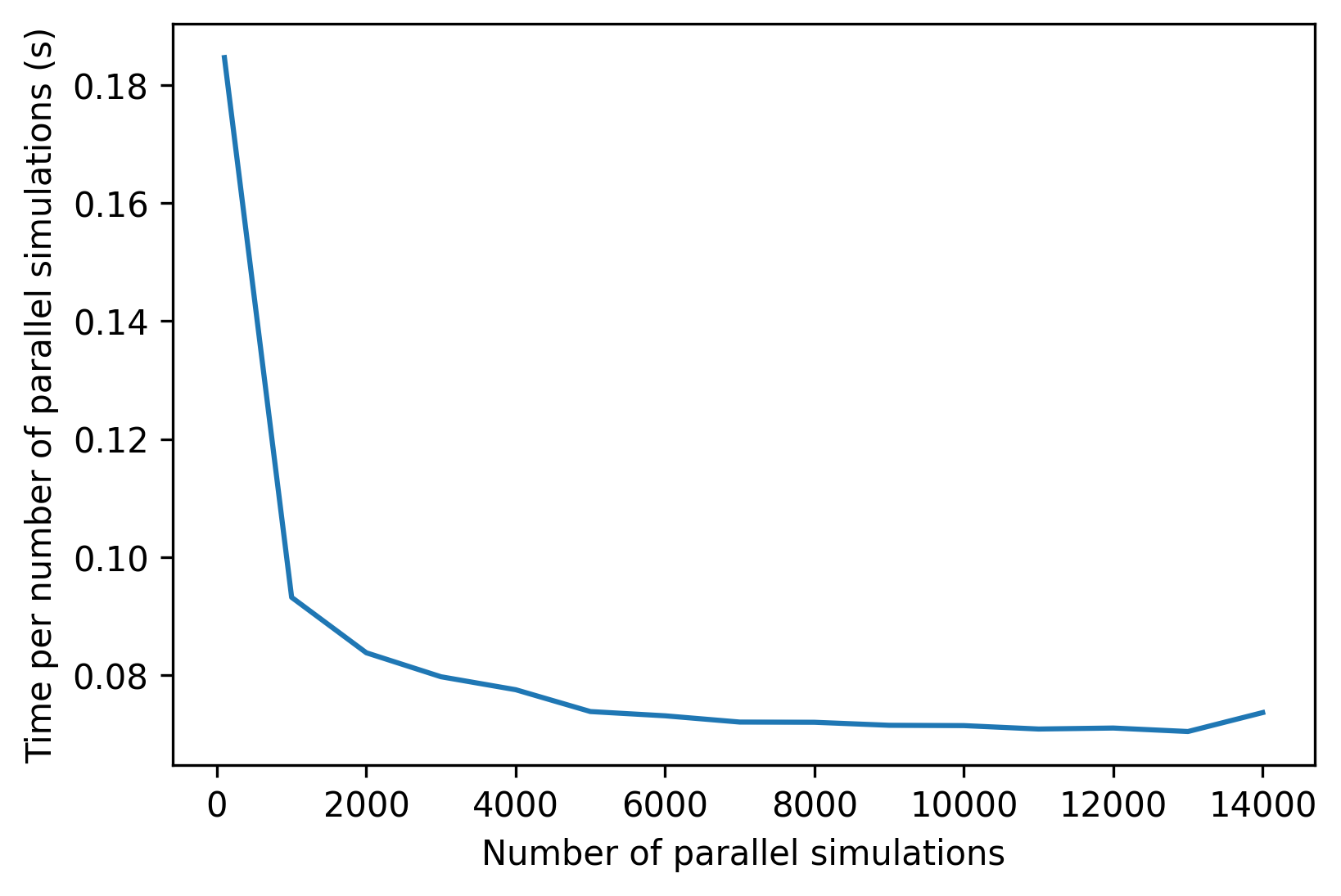}
\caption{For \code{odeint} with relative and absolute tolerances set to \code{1e-10}, the ratio of simulation time over number of parallel simulations is plotted, as a function of number of parallel simulations. For a given number of simulations, the input parameters are randomly generated (using the same random generator seed), and are simulated simultaneously via vectorization on an Nvidia A100 GPU with 80GB of memory. Above $6000$ simulations, the curve becomes nearly constant, indicating that the total simulation time is becoming a linear function of the number of simulations, which indicates that the parallelization capabilities of the GPU has become saturated.}
\label{figure:gpu_parallelization_saturation}
\end{figure}

Full perturbative solver GPU performance data is given in Figures \ref{figure:solver_data_gpu_full} and \ref{figure:solver_data_gpu_grad_full}. The corresponding data run on CPU are given in Figures \ref{figure:solver_data_cpu_full} and \ref{figure:solver_data_cpu_grad_full}. For the CPU data we reduce the number of random inputs to $100$. As we are only using a single core, there are no concerns about ensuring we have enough inputs to ``saturate'' parallelization of the device. Figure \ref{figure:speedup_cpu} presents the CPU speedup data. In the CPU case, we observe that for high accuracies, the perturbative solvers are actually slower than \code{odeint}. For medium range accuracies, we observe a more modest speedup than in the GPU case, with up to $7\times$ speedup for gradient computation. Despite the presence of some speedup, realistically this implementation of the perturbative solvers will not be useful on CPU, as the number of instances that can be run in parallel on multiple cores will be limited due to memory requirements. By contrast, standard ODE solvers have minimal memory requirements, and therefore can scale to be run in parallel on many more CPU cores, eliminating any practical advantage of the perturbative solvers. As described in the main text, the memory requirements could be alleviated by implementing the core loop of the perturbative solvers in a serial way. Future versions of these solvers including this feature could lead to better comparisons on CPU.

\begin{figure}[h!]
\includegraphics[scale=.6]{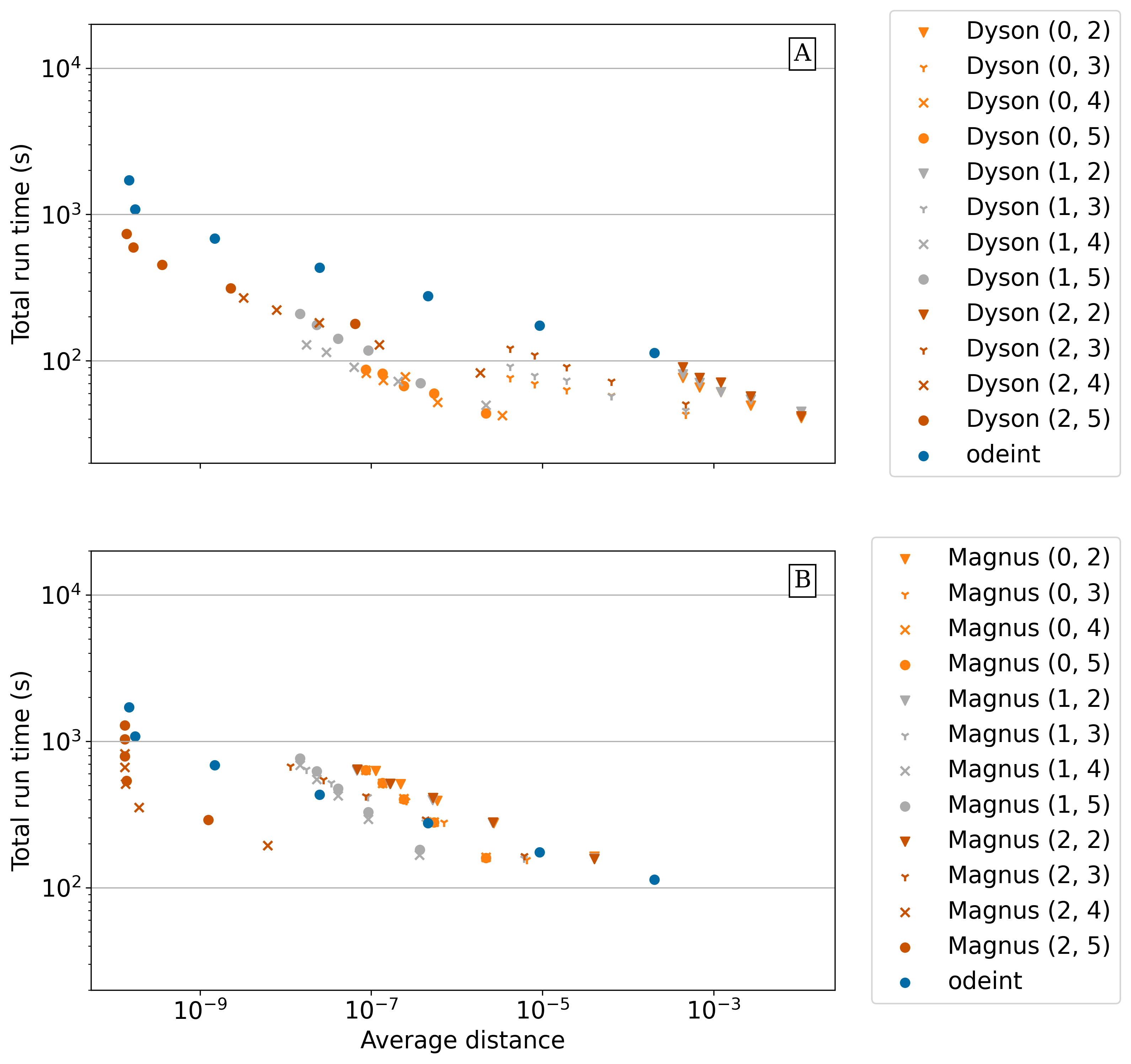}
\caption{Total Runtime v.s. Average Distance comparison on GPU for (A) Dyson-based perturbative solvers and (B) Magnus-based perturbative solvers. For both Dyson and Magnus solvers, the label $(d, n)$ denotes the configuration parameters used: a $d$-order Chebyshev approximation of the signals, and an $n^{th}$ order perturbative expansion. For a given solution $U$, distance is computed via the metric $\norm{U - V}_2 / \sqrt{d}$, where $V$ is a benchmark solution computed using \code{odeint} at the smallest possible tolerance values. Average distance is the arithmetic mean of these values. For each configuration, the data points correspond to various numbers of time-steps $M=10^4$, $2 \times 10^4$, $3 \times 10^4$, $4 \times 10^4$, or $5 \times 10^4$. The ODE solver was vectorized over all $7000$ inputs.}
\label{figure:solver_data_gpu_full}
\end{figure}

\begin{figure}[h!]
\includegraphics[scale=.6]{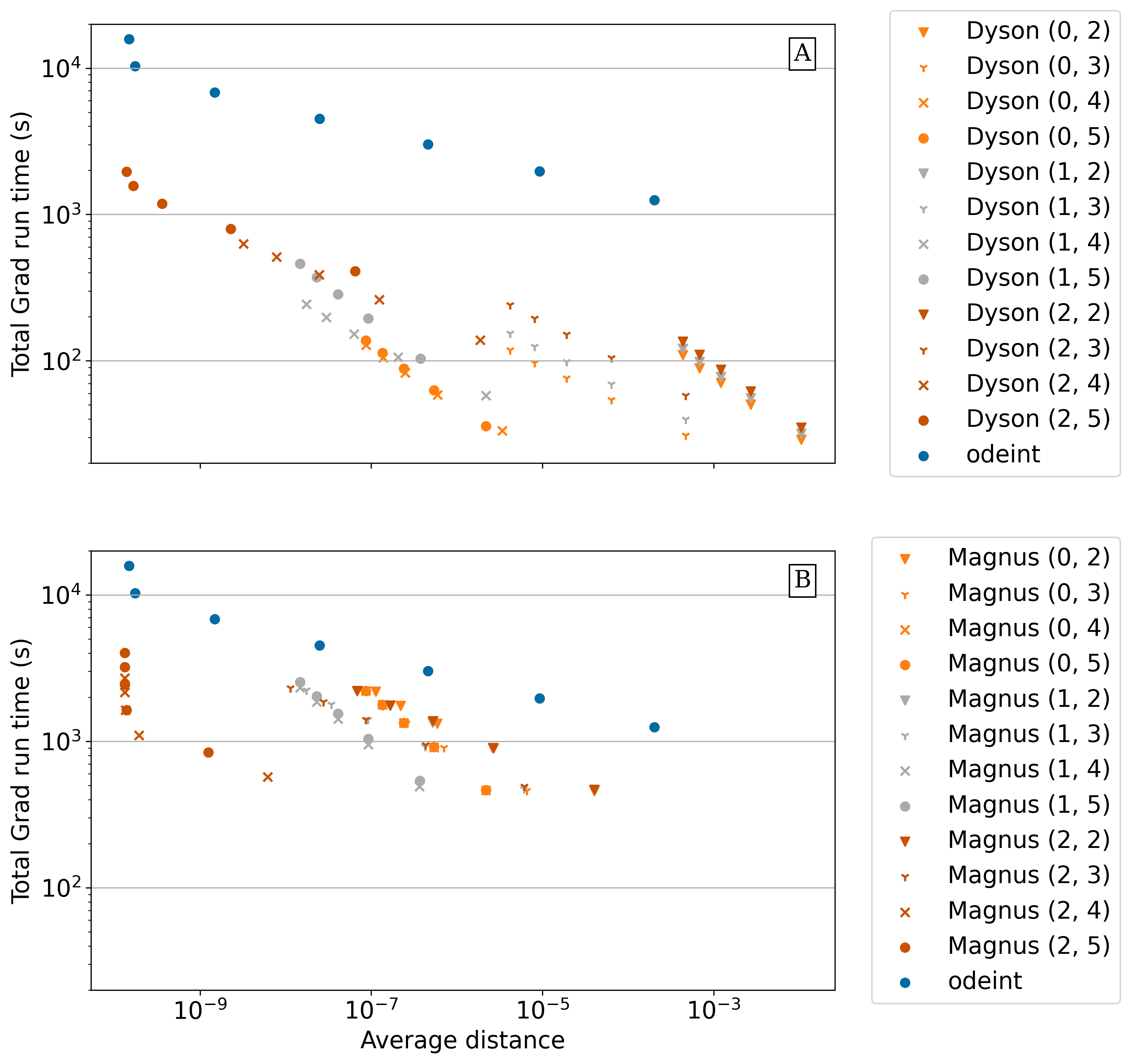}
\caption{Total Runtime v.s. Average Distance comparison for computing gradients on GPU for (A) Dyson-based perturbative solvers and (B) Magnus-based perturbative solvers. See the caption for Figure \ref{figure:solver_data_gpu_full} for further details on reading this plot.}
\label{figure:solver_data_gpu_grad_full}
\end{figure}

\begin{figure}[h!]
\includegraphics[scale=.6]{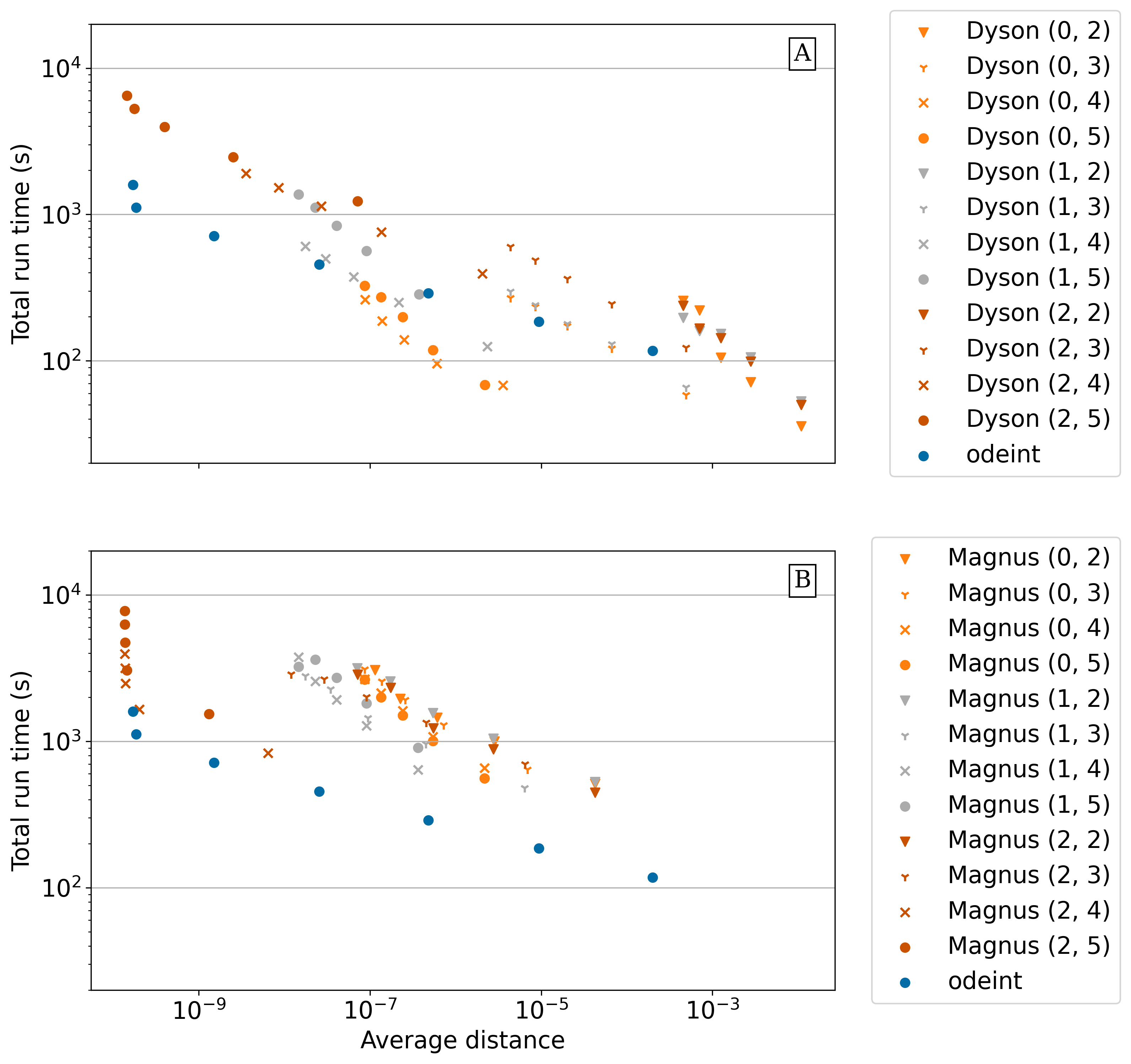}
\caption{Total Runtime v.s. Average Distance comparison on CPU for (A) Dyson-based perturbative solvers and (B) Magnus-based perturbative solvers. This data was generated using $100$ inputs run serially on a single-core CPU for all solvers. See the caption for Figure \ref{figure:solver_data_gpu_full} for a description of the legend.}
\label{figure:solver_data_cpu_full}
\end{figure}

\begin{figure}[h!]
\includegraphics[scale=.6]{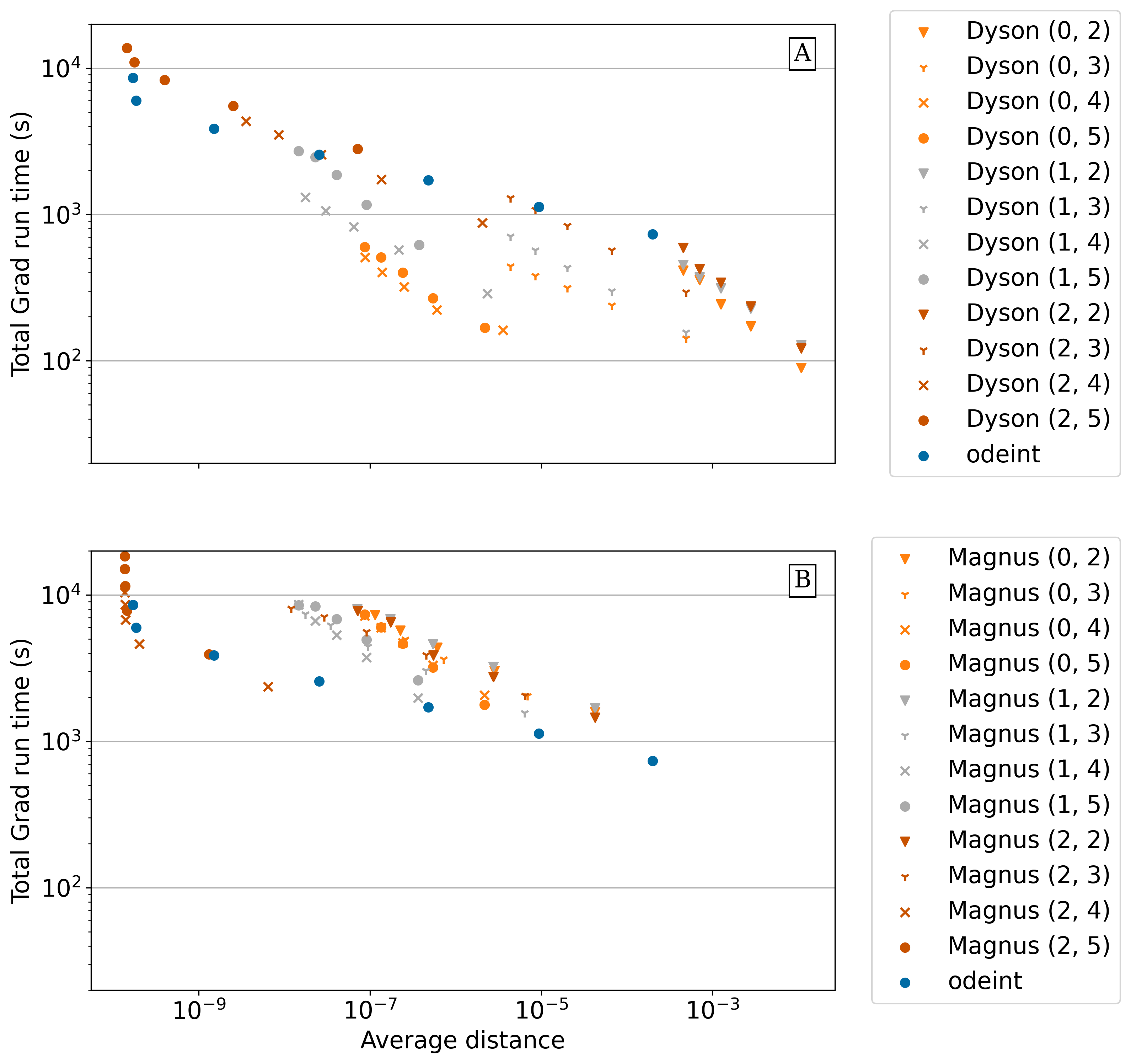}
\caption{Total Runtime v.s. Average Distance comparison for computing gradients on CPU for (A) Dyson-based perturbative solvers and (B) Magnus-based perturbative solvers. This data was generated using $100$ inputs run serially on a single-core CPU for all solvers. See the caption for Figure \ref{figure:solver_data_gpu_full} for a description of the legend.}
\label{figure:solver_data_cpu_grad_full}
\end{figure}

\clearpage

\begin{figure}[h!]
\centering
\hspace*{-1.25cm}
\includegraphics[scale=.45]{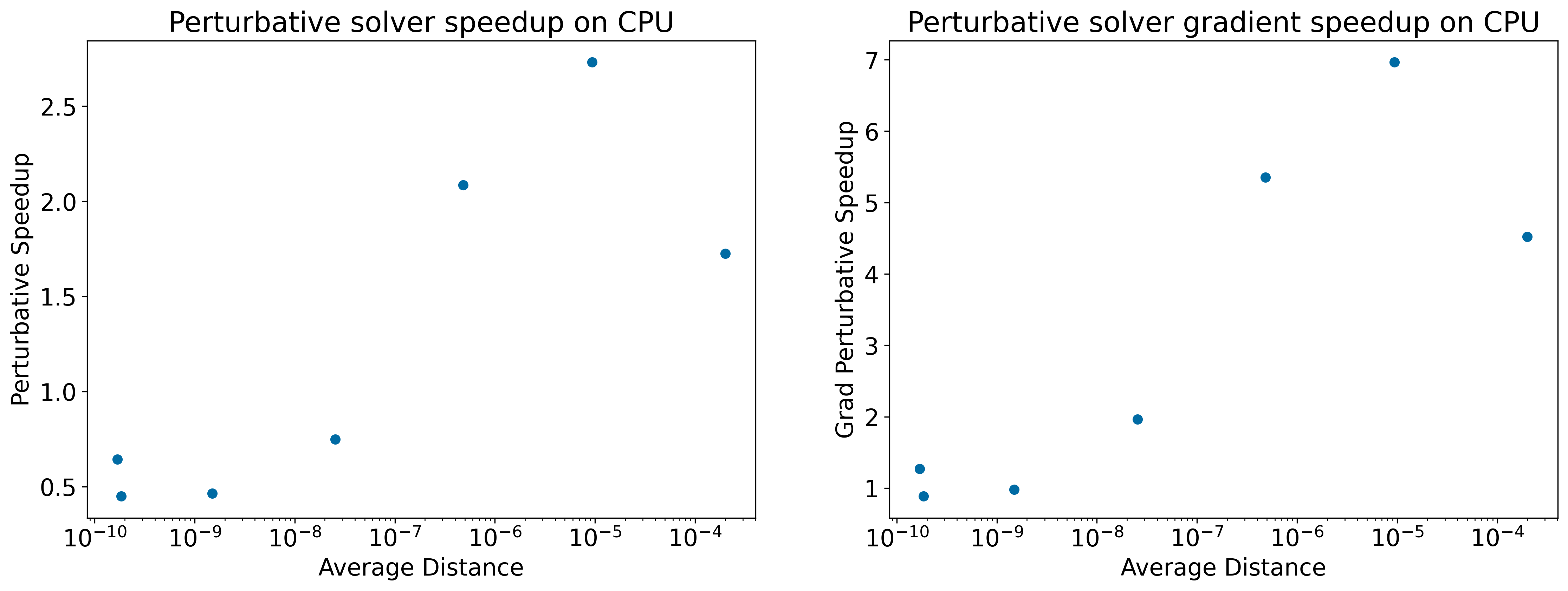}
\caption{Speedup vs average distance for the data shown in Figures \ref{figure:solver_data_cpu_full} and \ref{figure:solver_data_cpu_grad_full}. For each ODE solver point with the fastest time for a given average distance, we find the fastest perturbative solver point whose average distance is no larger than that of the ODE solver. The ``speedup'' is the ratio of the ODE solver speed over the perturbative solver speed.}
\label{figure:speedup_cpu}
\end{figure}
\section{Code examples} \label{app:code}

This appendix contains code blocks implemented with Qiskit Dynamics version $0.3.0$, demonstrating how to perform computations described in the main text. See the API documentation for the latest usage details.

\subsection{Constructing unitary approximations with the Magnus expansion} \label{app:code_solve_lmde_perturbation}

The following code block computes the unitary approximation of Section \ref{section:fidelity_approximation} using the \code{solve\_lmde\_perturbation} function. The pre-defined variables are:
\begin{itemize}
	\item \code{perturb\_list}: A list of callable functions for evaluating each perturbation.
	\item \code{T}: The time to integrate.
	\item \code{magnus\_order}: The order of the Magnus expansion to compute.
	\item \code{generator}: A callable function computing the unperturbed interaction frame generator, in this case the function evaluating $-iH_\emptyset(t)$.
	\item \code{c}: The perturbation coefficient values.
\end{itemize}
Additional package-specific information is supplied in the comments.

\begin{python}
# compute unperturbed unitary and Magnus expansion
# 'jax_odeint' is an integration method available in 
# Qiskit Dynamics
results = solve_lmde_perturbation(
	perturbations=perturb_list,
	t_span=[0, T],
	expansion_method='magnus',
	expansion_order=magnus_order,
	generator=generator,
	integration_method='jax_odeint',
	rtol=1e-12,
	atol=1e-12
)

# Extract unperturbed unitary from the results object
U = results.y[-1]

# construct ArrayPolynomial for Magnus expansion
term_labels = results.perturbation_results.expansion_labels
perturbation_terms = results.perturbation_results.expansion_terms
magnus_expansion = ArrayPolynomial(
	array_coefficients=perturbation_terms,
	monomial_labels=term_labels,
)

# compute approximate unitary for a given value of c
# jexpm is the JAX matrix exponentiation routine
U_approx = U @ jexpm(magnus_expansion(c).data)
\end{python}
The final line in the above code can be called repeatedly for different values of \code{c}.

\subsection{Robustness objective construction} \label{app:robustness_objective_code}

We walk through the computation of the $h_I(b)$ matrices described in Section \ref{section:robustness}. Let \code{mp} denote the \code{ArrayPolynomial} instance representing computed Magnus expansion truncation $\Omega(c, b)$, as constructed in the preceding section.

First, from $\Omega(c, b)$, we construct $\Omega(c, b)P - \frac{Tr(\Omega(c, b)P)}{2} P$. This can be done by interacting with an \code{ArrayPolynomial} using the same syntax as a standard \code{numpy} array:
\begin{python}
# project onto subspace by restricting to first two columns
mp = mp[:, 0:2]

# remove part of mp proportional to projection
mp_nontrivial =  (
	mp - mp.trace() * np.eye(5, 2, dtype=complex) / 2
)
\end{python}
In the above \code{mp\_nontrivial} is a \code{ArrayPolynomial} instance representing the array-valued polynomial given by $\Omega(c, b)P - \frac{Tr(\Omega(c, b)P)}{2} P$. 

Finally, we construct an \code{ArrayPolynomial} representing $h(c, b)$ (for fixed $b$) by applying the definition of the Frobenius norm to \code{mp\_nontrivial}, but when doing so, we only keep the terms $h_I(b)$ for which $\int_{c \in D} dc p(c) c_I \neq 0$. Let \code{multiset\_list} denote the list of \code{Multiset} instances representing the $I$ corresponding to the non-zero moments, and let \code{nonzero\_moments} denote the computed non-zero moments with the same ordering as \code{multiset\_list}. We compute $g(b)$ as follows (with the variable \code{robustness\_integral} corresponding to $g(b)$):
\begin{python}
# compute inner product of mp_nontrivial with itself
# keeping only terms corresponding to nonzero moments
mp_norm = mp_nontrivial.conj().mul(
	mp_nontrivial, 
	monomial_filter=lambda x: x in nonzero_multisets
).real.sum()

# compute integral by taking inner product of
# nonzero moments and the h_I coefficients
robustness_integral = np.dot(
	mp_norm.array_coefficients, nonzero_moments[1:]
).data
\end{python}

\subsection{Perturbative solver API} \label{app:perturbative_solver_api}

The classes \code{DysonSolver} and \code{MagnusSolver} implement the perturbative solvers described in Section \ref{section:numerical_integrators}. Both classes are instantiated with the following signature:
\begin{python}
solver = PerturbativeSolver(
	operators: List[Operator],
	rotating_frame: Union[Array, Operator, RotatingFrame, None],
	dt: float,
	carrier_freqs: Array,
	chebyshev_orders: List[int],
	expansion_order: Optional[int] = None,
	expansion_labels: Optional[List[Multiset]] = None,
	integration_method: Optional[str] = None,
	include_imag: Optional[List[bool]] = None,
	**kwargs
)
\end{python}
where:
\begin{itemize}
	\item \code{operators} is a list of the operators $A_j$.
	\item \code{rotating\_frame} is the frame operator $F$.
	\item \code{dt} is the step size $\Delta t$.
	\item \code{carrier\_freqs} is the list of analog carrier frequencies relative to which the expansion is computed. Note that these must be specified in frequency units, rather than angular frequency, as above.\footnote{I.e. a frequency $\nu_j$ specified in this list corresponds to $\omega_j = 2 \pi \nu_j$ in the mathematical presentation of the solver.}
	\item \code{chebyshev\_orders} is a list of integers indicating the number of Chebyshev terms to include when approximating each envelope, as in Equation \eqref{equation:env_approx}. A different value must be specified for each signal, with $0$ indicating a constant approximation.
	\item \code{expansion\_order} and \code{expansion\_labels} indicate the terms in the expansion to use when approximating the solution over each interval, with all terms up to \code{expansion\_order} being used, and any additional terms in \code{expansion\_labels} being included.
	\item \code{integration\_method} is the ODE method used to compute the perturbation terms; any method available via \code{qiskit\_dynamics.solve\_ode} can be used.
	\item \code{include\_imag} is a list of booleans indicating whether to keep, on a signal by signal basis, the second term in the decomposition in Equation \eqref{equation:real_imag_decomposition}.\footnote{This term may be dropped if $\omega_j=0$, or if there is some guarantee that $\Im[f_{j,m}e^{i\omega_jt_0}]=0$ over each interval startpoint. We note that, even if in the model there is a non-trivial carrier frequency for a given signal, it is possible to set $\omega_j=0$ and absorb the complex oscillator into the definition of the envelope. This will come at a potential cost of requiring smaller $\Delta t$ for accuracy, but at the benefit of fewer terms in the expansion.}
	\item \code{kwargs} are any additional keyword arguments to pass to \code{qiskit\_dynamics.solve\_ode} when computing the expansion.
\end{itemize}

The pre-computation of the expansion occurs at object instantiation. Once instantiated, the solver can be run on different envelopes by calling the \code{solve} method:
\begin{python}
yf = solver.solve(t0, n_steps, y0, signals)
\end{python}
where
\begin{itemize}
	\item \code{t0} and \code{n\_steps} specify the interval to solver over: $[\code{t0}, \code{t0 + n\_steps * dt}]$.
	\item \code{y0} is the initial state of the linear matrix differential equation.
	\item \code{signals} is a list of \code{qiskit\_dynamics.Signal} objects representing the time-dependent signals appearing in the decomposition of the generator. The carrier frequencies of these signals will be shifted to the reference carrier frequencies used at instantiation, with any leftover frequency absorbed into the envelope.
\end{itemize}
The method returns the final state \code{yf}.

While their usage is technical, \code{DysonSolver} and \code{MagnusSolver} enable usage of these methods while only requiring a mathematical description of the problem. Internally, they make use of the \code{solve\_lmde\_perturbation} function and \code{ArrayPolynomial} class.

\bibliography{dyson_magnus}

\begin{thebibliography}{10}

\bibitem{dyson_radiation_1949}
F.~J. Dyson.
\newblock The {Radiation} {Theories} of {Tomonaga}, {Schwinger}, and {Feynman}.
\newblock {\em Physical Review}, 75(3):486--502, February 1949.

\bibitem{magnus_exponential_1954}
W.~Magnus.
\newblock On the exponential solution of differential equations for a linear
  operator.
\newblock {\em Communications on Pure and Applied Mathematics}, 7(4):649--673,
  November 1954.

\bibitem{blanes_magnus_2009}
S.~Blanes, F.~Casas, J.~A. Oteo, and J.~Ros.
\newblock The {Magnus} expansion and some of its applications.
\newblock {\em Physics Reports}, 470(5–6):151--238, January 2009.

\bibitem{soliverez_general_1981}
C.~E. Soliverez.
\newblock General theory of effective {Hamiltonians}.
\newblock {\em Physical Review A}, 24(1):4--9, 1981.
\newblock Publisher: American Physical Society.

\bibitem{james_effective_2011}
D.~F. James and J.~Jerke.
\newblock Effective {Hamiltonian} theory and its applications in quantum
  information.
\newblock {\em Canadian Journal of Physics}, 2011.
\newblock Publisher: NRC Research Press Ottawa, Canada.

\bibitem{bravyi_schriefferwolff_2011}
S.~Bravyi, D.~P. DiVincenzo, and D.~Loss.
\newblock Schrieffer–{Wolff} transformation for quantum many-body systems.
\newblock {\em Annals of Physics}, 326(10):2793--2826, 2011.

\bibitem{zeuch_exact_2020}
D.~Zeuch, F.~Hassler, J.~J. Slim, and D.~P. DiVincenzo.
\newblock Exact rotating wave approximation.
\newblock {\em Annals of Physics}, 423:168327, December 2020.

\bibitem{gambetta_analytic_2011}
J.~M. Gambetta, F.~Motzoi, S.~T. Merkel, and F.~K. Wilhelm.
\newblock Analytic control methods for high-fidelity unitary operations in a
  weakly nonlinear oscillator.
\newblock {\em Physical Review A}, 83(1):012308, 2011.
\newblock Publisher: American Physical Society.

\bibitem{magesan_effective_2020}
E.~Magesan and J.~M. Gambetta.
\newblock Effective {Hamiltonian} models of the cross-resonance gate.
\newblock {\em Physical Review A}, 101(5):052308, 2020.
\newblock Publisher: American Physical Society.

\bibitem{malekakhlagh_first-principles_2020}
M.~Malekakhlagh, E.~Magesan, and D.~C. McKay.
\newblock First-principles analysis of cross-resonance gate operation.
\newblock {\em Physical Review A}, 102(4):042605, 2020.
\newblock Publisher: American Physical Society.

\bibitem{haeberlen_1968}
U.~Haeberlen and J.~S. Waugh.
\newblock Coherent {A}veraging {E}ffects in {M}agnetic {R}esonance.
\newblock {\em Physical Review}, 175:453--467, 1968.

\bibitem{waugh_1968}
J.~S. Waugh, L.~M. Huber, and U.~Haeberlen.
\newblock Approach to {H}igh-{R}esolution nmr in solids.
\newblock {\em Physical Review Letters}, 20:180--182, 1968.

\bibitem{mansfield_1971}
P.~Mansfield.
\newblock Symmetrized pulse sequences in high resolution {NMR} in solids.
\newblock {\em Journal of Physics C: Solid State Physics}, 4(11):1444, 1971.

\bibitem{rhim_1973}
W.-K. Rhim, D.~D. Elleman, and R.~W. Vaughan.
\newblock Analysis of multiple pulse {NMR} in solids.
\newblock {\em The Journal of Chemical Physics}, 59(7):3740--3749, 1973.

\bibitem{mehring_1983}
M.~Mehring.
\newblock {\em Principles of {High} {Resolution} {NMR} in {Solids}}.
\newblock Springer-Verlag Berlin Heidelberg, 1983.

\bibitem{takegoshi_1985}
K.~Takegoshi and C.~A. McDowell.
\newblock A “magic echo” pulse sequence for the high-resolution {NMR}
  spectra of abundant spins in solids.
\newblock {\em Chemical Physics Letters}, 116(2):100 -- 104, 1985.

\bibitem{levitt_1986}
M.~H. Levitt.
\newblock Composite pulses.
\newblock {\em Progress in Nuclear Magnetic Resonance Spectroscopy}, 18(2):61
  -- 122, 1986.

\bibitem{cory_1991}
D.~G. Cory.
\newblock A new multiple-pulse cycle for homonuclear dipolar decoupling.
\newblock {\em Journal of Magnetic Resonance (1969)}, 94(3):526 -- 534, 1991.

\bibitem{cory_1990}
D.~G. Cory, J.~B. Miller, R.~Turner, and A.~N. Garroway.
\newblock Multiple-pulse methods of {1H N.M.R.} imaging of solids:
  second-averaging.
\newblock {\em Molecular Physics}, 70(2):331--345, 1990.

\bibitem{viola_1998}
L.~Viola and S.~Lloyd.
\newblock Dynamical suppression of decoherence in two-state quantum systems.
\newblock {\em Physical Review A}, 58:2733--2744, 1998.

\bibitem{viola_1999}
L.~Viola, S.~Lloyd, and E.~Knill.
\newblock Universal {C}ontrol of {D}ecoupled {Q}uantum {S}ystems.
\newblock {\em Physical Review Letters}, 83:4888--4891, 1999.

\bibitem{khodjasteh_2005}
K.~Khodjasteh and D.~A. Lidar.
\newblock Fault-{T}olerant {Q}uantum {D}ynamical {D}ecoupling.
\newblock {\em Physical Review Letters}, 95:180501, 2005.

\bibitem{khodjasteh_2009}
K.~Khodjasteh and L.~Viola.
\newblock Dynamically {E}rror-{C}orrected {G}ates for {U}niversal {Q}uantum
  {C}omputation.
\newblock {\em Physical Review Letters}, 102:080501, 2009.

\bibitem{khodjasteh_2010}
K.~Khodjasteh, D.~A. Lidar, and L.~Viola.
\newblock Arbitrarily {A}ccurate {D}ynamical {C}ontrol in {O}pen {Q}uantum
  {S}ystems.
\newblock {\em Physical Review Letters}, 104:090501, 2010.

\bibitem{green_2013}
T.~J. Green, J.~Sastrawan, H.~Uys, and M.~J. Biercuk.
\newblock Arbitrary quantum control of qubits in the presence of universal
  noise.
\newblock {\em New Journal of Physics}, 15(9):095004, 2013.

\bibitem{pazsilva_2014}
G.~A. Paz-Silva and L.~Viola.
\newblock General {T}ransfer-{F}unction {A}pproach to {N}oise {F}iltering in
  {O}pen-{L}oop {Q}uantum {C}ontrol.
\newblock {\em Physical Review Letters}, 113:250501, 2014.

\bibitem{ribeiro_systematic_2017}
H.~Ribeiro, A.~Baksic, and A.~A. Clerk.
\newblock Systematic {Magnus}-{Based} {Approach} for {Suppressing} {Leakage}
  and {Nonadiabatic} {Errors} in {Quantum} {Dynamics}.
\newblock {\em Physical Review X}, 7(1):011021, February 2017.
\newblock Publisher: American Physical Society.

\bibitem{evans_timedependent_1967}
W.~A.~B. Evans and J.~G. Powles.
\newblock A time-dependent {Dyson} expansion - the nuclear resonance signal in
  a rotating single crystal.
\newblock {\em Proceedings of the Physical Society}, 92(4):1046--1054, 1967.

\bibitem{haas_engineering_2019}
H.~Haas, D.~Puzzuoli, F.~Zhang, and D.~G Cory.
\newblock Engineering effective {Hamiltonians}.
\newblock {\em New Journal of Physics}, 21(10):103011, October 2019.

\bibitem{tabatabaei_numerical_2020}
S.~Tabatabaei, H.~Haas, W.~Rose, B.~Yager, M.~Piscitelli, P.~Sahafi, A.~Jordan,
  P.~J. Poole, D.~Dalacu, and R.~Budakian.
\newblock Numerical {Engineering} of {Robust} {Adiabatic} {Operations}.
\newblock {\em Physical Review Applied}, 15(4):044043, April 2021.
\newblock Publisher: American Physical Society.

\bibitem{cywinski_how_2008}
Ł. Cywiński, R.~M. Lutchyn, C.~P. Nave, and S.~Das~Sarma.
\newblock How to enhance dephasing time in superconducting qubits.
\newblock {\em Physical Review B}, 77(17):174509, 2008.
\newblock Publisher: American Physical Society.

\bibitem{green_arbitrary_2013}
T.~J. Green, J.~Sastrawan, H.~Uys, and M.~J. Biercuk.
\newblock Arbitrary quantum control of qubits in the presence of universal
  noise.
\newblock {\em New Journal of Physics}, 15(9):095004, 2013.

\bibitem{hangleiter_filter-function_2021}
T.~Hangleiter, P.~Cerfontaine, and H.~Bluhm.
\newblock Filter-function formalism and software package to compute quantum
  processes of gate sequences for classical non-{Markovian} noise.
\newblock {\em Physical Review Research}, 3(4):043047, 2021.
\newblock Publisher: American Physical Society.

\bibitem{ball_software_2021}
H.~Ball, M.~Biercuk, A.~Carvalho, J.~Chen, M.~R. Hush, L.~A.~de Castro, L.~Li,
  P.~J. Liebermann, H.~Slatyer, C.~Edmunds, V.~Frey, C.~Hempel, and A.~Milne.
\newblock Software tools for quantum control: {Improving} quantum computer
  performance through noise and error suppression.
\newblock {\em Quantum Science and Technology}, 2021.

\bibitem{tobias_hangleiter_2022_6574304}
T.~Hangleiter, J.~D. Teske, and H.~Bluhm.
\newblock qutech/filter\_functions: Release 1.1.2, May 2022.

\bibitem{shillito_fast_2020}
R.~Shillito, J.~A. Gross, A.~Di~Paolo, É. Genois, and A.~Blais.
\newblock Fast and differentiable simulation of driven quantum systems.
\newblock {\em Physical Review Research}, 3(3):033266, September 2021.
\newblock Publisher: American Physical Society.

\bibitem{kalev_integral-free_2020}
A.~Kalev and I.~Hen.
\newblock An integral-free representation of the {Dyson} series using divided
  differences.
\newblock {\em arXiv:2010.09888}, October 2020.

\bibitem{wiki_multi_index_2021}
Multi-index notation, September 2021.
\newblock URL:https://en.wikipedia.org/wiki/Multi-index\_notation, Page Version
  ID: 1046377330.

\bibitem{time_ordering_2022}
Path-ordering, October 2022.
\newblock URL:https://en.wikipedia.org/wiki/Path-ordering, Page Version ID:
  1116280769.

\bibitem{vanloan_1978}
C.~Van~Loan.
\newblock Computing integrals involving the matrix exponential.
\newblock {\em {IEEE} Transactions on Automatic Control}, 23(3):395--404, 1978.

\bibitem{carbonell_2008}
F.~Carbonell, J.~C. J{\'i}menez, and L.~M. Pedroso.
\newblock Computing multiple integrals involving matrix exponentials.
\newblock {\em Journal of Computational and Applied Mathematics},
  213(1):300--305, 2008.

\bibitem{goodwin_2015}
D.~L. Goodwin and I.~Kuprov.
\newblock Auxiliary matrix formalism for interaction representation
  transformations, optimal control, and spin relaxation theories.
\newblock {\em The Journal of Chemical Physics}, 143(8):084113, 2015.

\bibitem{machnes_2018}
S.~Machnes, E.~Ass\'emat, D.~Tannor, and F.~K. Wilhelm.
\newblock Tunable, {F}lexible, and {E}fficient {O}ptimization of {C}ontrol
  {P}ulses for {P}ractical {Q}ubits.
\newblock {\em Physical Review Letters}, 120:150401, 2018.

\bibitem{hindmarsh_sundials_2005}
A.~C. Hindmarsh, P.~N. Brown, K.~E. Grant, S.~L. Lee, R.~Serban, D.~E.
  Shumaker, and C.~S. Woodward.
\newblock {SUNDIALS}: {Suite} of nonlinear and differential/algebraic equation
  solvers.
\newblock {\em ACM Transactions on Mathematical Software}, 31(3):363--396,
  2005.

\bibitem{rackauckas_comparison_2018}
C.~Rackauckas, Y.~Ma, V.~Dixit, X.~Guo, M.~Innes, J.~Revels, J.~Nyberg, and
  V.~Ivaturi.
\newblock A {Comparison} of {Automatic} {Differentiation} and {Continuous}
  {Sensitivity} {Analysis} for {Derivatives} of {Differential} {Equation}
  {Solutions}.
\newblock {\em arXiv:1812.01892 [cs]}, 2018.

\bibitem{burum_magnus_1981}
D.~P. Burum.
\newblock Magnus expansion generator.
\newblock {\em Physical Review B}, 24(7):3684--3692, October 1981.
\newblock Publisher: American Physical Society.

\bibitem{salzman_alternative_1985}
W.~R. Salzman.
\newblock An alternative to the magnus expansion in time‐dependent
  perturbation theory.
\newblock {\em The Journal of Chemical Physics}, 82(2):822--826, January 1985.
\newblock Publisher: American Institute of Physics.

\bibitem{arnal_general_2018}
A.~Arnal, F.~Casas, and C.~Chiralt.
\newblock A general formula for the {Magnus} expansion in terms of iterated
  integrals of right-nested commutators.
\newblock {\em Journal of Physics Communications}, 2(3):035024, March 2018.
\newblock Publisher: IOP Publishing.

\bibitem{225963}
J.~Silverman.
\newblock Number of polynomial terms for certain degree and certain number of
  variables.
\newblock MathOverflow.
\newblock URL:https://mathoverflow.net/q/225963 (version: 2015-12-12).

\bibitem{qiskit_dynamics_2021}
D.~J. Egger, H.~Landa, A.~Parr, D.~Puzzuoli, B.~Rosand, R.~K. Rupesh,
  M.~Treinish, and C.~J. Wood.
\newblock Qiskit {D}ynamics, 2021.

\bibitem{multiset}
multiset.
\newblock URL:https://pypi.org/project/multiset/.

\bibitem{jax2018github}
J.~Bradbury, R.~Frostig, P.~Hawkins, M.~J. Johnson, C.~Leary, D.~Maclaurin,
  G.~Necula, A.~Paszke, J.~Vander{P}las, S.~Wanderman-{M}ilne, and Q.~Zhang.
\newblock {JAX}: composable transformations of {P}ython+{N}um{P}y programs,
  2018.

\bibitem{supplemental_repo}
D.~Puzzuoli, S.~Lin, M.~Malekakhlagh, E.~Pritchett, B.~Rosand, and C.~J. Wood.
\newblock Supplemental code repository, 2022.
\newblock URL:https://github.com/DanPuzzuoli/multivariable\_dyson\_magnus.

\bibitem{al-mohy_new_2009}
A.~Al-Mohy and N.~Higham.
\newblock A {New} {Scaling} and {Squaring} {Algorithm} for the {Matrix}
  {Exponential}.
\newblock {\em SIAM Journal on Matrix Analysis and Applications},
  31(3):970--989, August 2009.

\bibitem{shampine_practical_1986}
Lawrence~F. Shampine.
\newblock Some practical {Runge}-{Kutta} formulas.
\newblock {\em Mathematics of Computation}, 46(173):135--150, 1986.

\bibitem{johansson_qutip_2013}
J.~R. Johansson, P.~D. Nation, and Franco Nori.
\newblock {QuTiP} 2: {A} {Python} framework for the dynamics of open quantum
  systems.
\newblock {\em Computer Physics Communications}, 184(4):1234--1240, 2013.

\bibitem{Jurcevic_2021}
P.~Jurcevic, A.~Javadi-Abhari, L.~S. Bishop, I.~Lauer, D.~F. Bogorin, M.~Brink,
  L.~Capelluto, O.~Günlük, T.~Itoko, N.~Kanazawa, A.~Kandala, G.~A. Keefe,
  K.~Krsulich, W.~Landers, E.~P. Lewandowski, D.~T. McClure, G.~Nannicini,
  A.~Narasgond, H.~M. Nayfeh, E.~Pritchett, M.~B. Rothwell, S.~Srinivasan,
  N.~Sundaresan, C.~Wang, K.~X. Wei, C.~J. Wood, J.-B. Yau, E.~J. Zhang, O.~E.
  Dial, J.~M. Chow, and J.~M. Gambetta.
\newblock Demonstration of quantum volume 64 on a superconducting quantum
  computing system.
\newblock {\em Quantum Science and Technology}, 6(2):025020, 2021.
\newblock Publisher: IOP Publishing.

\end{thebibliography}
\bibliographystyle{unsrt}
\end{document}